\documentclass[11pt,a4paper]{article}
\usepackage[english]{babel}
\usepackage{german}
\usepackage{fullpage}
\usepackage{graphicx}
\usepackage{subfigure}
\usepackage{algorithm}
\usepackage[noend]{algorithmic}
\usepackage{amsmath,amssymb,amsfonts,cite}
\usepackage{footmisc}

\newtheorem{theorem}{Theorem}

\newtheorem{corollary}{Corollary}

\newtheorem{definition}{Definition}

\newtheorem{observation}{Observation}

\newtheorem{lemma}{Lemma}
\newtheorem{notation}{Notation}

\newtheorem{remark}{Remark}

\newenvironment{proof}[1][Proof]{\noindent\textbf{#1.} }{\ \rule{0.5em}{0.5em}}
   \renewenvironment{thebibliography}[1]{
 \begin{oldthebibliography}{#1}
 \setlength{\parskip}{-0.1ex} \setlength{\itemsep}{0.78ex}}  {\end{oldthebibliography}}

\begin{document}

\title{A Simple Polynomial Algorithm for the \\
Longest Path Problem on Cocomparability Graphs}
\author{George B. Mertzios\thanks{%
Department of Computer Science, RWTH Aachen University, Aachen, Germany. 
Email:~\texttt{mertzios@cs.rwth-aachen.de}} 
\and Derek G. Corneil\thanks{%
Department of Computer Science, University of Toronto, Toronto, Canada.
Email: \texttt{dgc@cs.utoronto.ca}} \thanks{%
The second author wishes to thank the Natural Sciences and Engineering
Research Council of Canada for financial assistance.\vspace{-0.5cm}}}
\date{\vspace{-0.7cm}}
\maketitle

\begin{abstract}
Given a graph $G$, the longest path problem asks to compute a simple path of 
$G$ with the largest number of vertices. This problem is the most natural
optimization version of the well known and well studied Hamiltonian path
problem, and thus it is NP-hard on general graphs. However, in contrast to
the Hamiltonian path problem, there are only few restricted graph families
such as trees and some small graph classes where polynomial algorithms for
the longest path problem have been found. Recently it has been shown that
this problem can be solved in polynomial time on interval graphs by applying
dynamic programming to a characterizing ordering of the vertices of the
given graph~\cite{longest-int-algo}, thus answering an open question. In the
present paper, we provide the first polynomial algorithm for the longest
path problem on a much greater class, namely on cocomparability graphs. Our
algorithm uses a similar --~but essentially simpler~-- dynamic programming
approach, which is applied to a \emph{Lexicographic Depth First Search (LDFS)%
} characterizing ordering of the vertices of a cocomparability graph.
Therefore, our results provide evidence that this general dynamic
programming approach can be used in a more general setting, leading to
efficient algorithms for the longest path problem on greater classes of
graphs. LDFS has recently been introduced in~\cite{Corneil-LDFS08}. Since
then, a similar phenomenon of extending an existing interval graph algorithm
to cocomparability graphs by using an LDFS preprocessing step has also been
observed for the minimum path cover problem~\cite{Corneil-MPC}. Therefore,
more interestingly, our results also provide evidence that cocomparability
graphs present an interval graph structure when they are considered using 
an LDFS ordering of their vertices, which may lead to other new and more
efficient combinatorial algorithms.\newline

\noindent \textbf{Keywords:} Cocomparability graphs, longest path problem,
lexicographic depth first search, dynamic programming.
\end{abstract}

\section{Introduction\label{sec:intro}}

The Hamiltonian path problem, i.e.~the problem of deciding whether a graph $G
$ contains a simple path that visits each vertex of $G$ exactly once, is one
of the most well known NP-complete problems, with numerous applications. The
most natural optimization version of this problem is the \emph{longest path}
problem, that is, to compute a simple path of maximum length, or
equivalently, to find a maximum induced subgraph which is Hamiltonian. Even
if a graph itself is not Hamiltonian, it makes sense in several applications
to search for a longest path. However, computing a longest path seems to be
more difficult than deciding whether or not a graph admits a Hamiltonian
path. Indeed, it has been proved that even if a graph is Hamiltonian, the
problem of computing a path of length $n-n^{\varepsilon}$ for any $%
\varepsilon <1$ is NP-hard, where $n$ is the number of vertices of the input
graph~\cite{Karger97}. Moreover, there is no polynomial-time constant-factor
approximation algorithm for the longest path problem unless P=NP~\cite{Karger97}.

The Hamiltonian path problem (as well as many of its variants, e.g.~the Hamiltonian cycle problem) 
is known to be NP-complete on general graphs~\cite{GaJoTarj79}; 
furthermore, it remains NP-complete even when the input is restricted to
some small classes of graphs, such as split graphs~\cite{Golumbic04},
chordal bipartite graphs, split strongly chordal graphs~\cite{Muller96},
directed path graphs~\cite{Naras89}, circle graphs~\cite{Damasc89}, planar
graphs~\cite{GaJoTarj79}, and grid graphs~\cite{ItaiPap82}. 
On other restricted families of graphs, however, considerable success has been 
achieved in finding polynomial time algorithms for the Hamiltonian path problem. 
In particular, this problem can be solved polynomially on proper interval
graphs~\cite{Bertossi83}, interval graphs~\cite{Arikati90,Keil85}, 
and cocomparability graphs~\cite{Damasc92} 
(see~\cite{Brandstaedt99,Golumbic04} for definitions of these and other graph classes mentioned in the paper).

In contrast to the Hamiltonian path problem, there are only a few known
polynomial algorithms for the longest path problem, and, until recently, 
these were restricted to trees~\cite{Bult02}, weighted trees and block 
graphs~\cite{Uehara04}, bipartite permutation graphs~\cite{UeharaValiente07}, 
and ptolemaic graphs~\cite{TakaharaUehara08}. 
In~\cite{Uehara04} the question was raised whether the problem could be
solved in polynomial time for a much larger class, namely interval graphs.
Very recently such an algorithm has been discovered~\cite{longest-int-algo}. 
This algorithm applies dynamic programming to the vertex 
ordering of the given interval graph that is obtained after sorting the
intervals according to their right endpoints. 
Another natural generalization of the Hamiltonian path problem is 
the \emph{minimum path cover} problem, where the goal is to cover each vertex 
of the graph exactly once using the smallest number of simple paths. 
Clearly a solution to either the longest path or the minimum path cover problems 
immediately yields a solution to the Hamiltonian path problem.
Unlike the situation for the longest path problem, polynomial time
algorithms for the Hamiltonian cycle~\cite{Keil85}, Hamiltonian path, and
minimum path cover~\cite{Arikati90} problems on interval graphs have been
available since the 1980s and early 1990s.

Cocomparability graphs (i.e.~graphs, the complements of which can be
transitively oriented) strictly contain interval and permutation graphs~\cite{Brandstaedt99} 
and have also been studied with respect to various Hamiltonian problems. 
In particular, it is well known that the Hamiltonian path and Hamiltonian cycle
problems~\cite{Deogun94}, as well as the minimum path cover problem
(referred to as the Hamiltonian path completion problem~\cite{Damasc92}) are
polynomially solvable on cocomparability graphs. On the other hand, the
complexity status of the longest path problem on cocomparability graphs
--~and even on the smaller class of permutation graphs~-- has long been open.

Until recently, the only polynomial algorithms for the Hamiltonian path,
Hamiltonian cycle~\cite{Deogun94}, and minimum path cover~\cite{Damasc92}
problems on cocomparability graphs exploited the relationship between these
problems and the bump number of a poset representing the transitive
orientation of the complement graph. 
Furthermore, it had long been an open question whether there are algorithms
for these problems that, as with the interval graph algorithms, are based on
the structure of cocomparability graphs. This question has been recently
answered by the algorithm in~\cite{Corneil-MPC},
which solves the minimum path cover problem on cocomparability graphs by
building off the corresponding algorithm for interval graphs~\cite{Arikati90}
and using a preprocessing step based on the recently discovered \emph{%
Lexicographic Depth First Search (LDFS)}~\cite{Corneil-LDFS08}.

In the present paper we provide the first polynomial algorithm for the
longest path problem on cocomparability graphs (and thus also on permutation
graphs). 
Our algorithm develops a similar --~but much simpler~-- dynamic programming
approach to that of~\cite{longest-int-algo}, which is applied to a \emph{%
Lexicographic Depth First Search (LDFS)} characterizing ordering of the
vertices of a cocomparability graph (see~\cite{Corneil-LDFS08,Corneil-MPC}).
As a byproduct, this algorithm solves also the longest path problem on
interval graphs in a much simpler way than that of~\cite{longest-int-algo}
(the algorithm of~\cite{longest-int-algo} consists of three phases, during
which it introduces several dummy vertices to construct a second auxiliary
graph). Furthermore, these results provide evidence that this general
dynamic programming approach can be used in a more general setting,
providing efficient algorithms for the longest path problem on greater
classes of graphs. As already mentioned above, a similar phenomenon of
extending an existing interval graph algorithm to cocomparability graphs by
using an LDFS preprocessing step has also been observed for the minimum path
cover problem~\cite{Corneil-MPC}. Therefore, more interestingly, our results
also provide evidence that cocomparability graphs present an interval graph
structure when they are considered using an LDFS ordering of their vertices,
which may lead to other new and more efficient combinatorial algorithms.

\vspace{-0.2cm}
\paragraph{Organization of the paper}

In Section~\ref{prelim} we provide the necessary preliminaries and notation,
including vertex ordering characterizations of interval graphs,
cocomparability graphs, and LDFS orderings. In Section~\ref{sec:normal} we
study the effect of an LDFS preprocessing step on the vertex ordering
characterization of cocomparability graphs. This section provides much of
the structural foundation for our longest path algorithm that is presented
in Section~\ref{sec:longest}. Finally, we discuss the presented results and
further research in Section~\ref{conclusions}.

\section{Preliminaries and notation\label{prelim}}

In this article we follow standard notation and terminology, see for
instance~\cite{Golumbic04}. We consider finite, simple, and undirected graphs.
Given a graph $G=(V,E)$, we denote by $n$ the cardinality of $V$. An edge
between vertices $u$ and $v$ is denoted by $uv$, and in this case vertices $u
$ and $v$ are said to be \emph{adjacent}. $\overline{G}$ denotes the \emph{%
complement} of $G$, i.e.~$\overline{G}=(V,\overline{E})$, where $uv\in 
\overline{E}$ if and only if $uv\notin E$. Let $S\subseteq V$ be a set of
vertices of~$G$. Then, the cardinality of $S$ is denoted by~$|S|$ and the
subgraph of $G$ \emph{induced} by $S$ is denoted by $G[S]$, i.e.~$G[S]=(S,F)$%
, where for any two vertices $u,v\in S$, $uv\in F$ if and only if $uv\in E$.
The set ${N(v)=\{u \in V \ | \ uv \in E\}}$ is called the \emph{neighborhood}
of the vertex $v \in V$ in $G$.

A \emph{simple path} $P$ of a graph $G$ is a sequence of distinct vertices $%
v_{1}, v_{2}, \ldots, v_{k}$ such that~${v_{i}v_{i+1} \in E}$, for each $%
i\in\{1,2,\ldots,k-1\}$, and is denoted by $P=(v_{1}, v_{2}, \ldots, v_{k})$%
; throughout the paper all paths considered are simple. Furthermore, $v_{1}$
(resp.~$v_{k}$) is called the \emph{first} (resp.~\emph{last}) vertex of $P$. 
We denote by $V(P)$ the set of vertices of the path $P$, and define the 
\emph{length} $|P|$ of $P$ to be the number of vertices in $P$, i.e.~${%
|P|=|V(P)|}$. Additionally, if~${P=(v_{1},v_{2},\ldots,v_{i-1},v_{i},%
\ldots,v_{j},v_{j+1},\ldots,v_{k})}$ is a path of a graph and ${%
P_{0}=(v_{i},\ldots,v_{j})}$ is a subpath of $P$, we sometimes equivalently
use the notation~${P=(v_{1},v_{2},\ldots,v_{i-1},P_{0},v_{j+1},\ldots,v_{k})}$.

Recall that \emph{interval} graphs are the intersection graphs of closed
intervals on the real line. Furthermore, a \emph{comparability} graph is a
graph whose edges can be transitively oriented (i.e.~if~$x \rightarrow y$ and~$y
\rightarrow z$, then~$x \rightarrow z$); a \emph{cocomparability} graph $G$
is a graph whose complement~$\overline{G}$ is a comparability graph. \emph{%
Permutation} graphs are exactly the intersection of comparability and
cocomparability graphs. Moreover, cocomparability graphs strictly contain
interval graphs and permutation graphs, as well as other families of graphs
such as trapezoid graphs and cographs~\cite{Brandstaedt99}. 

\subsection{Vertex ordering characterizations\label{orderings-subsection}}

We now state vertex ordering characterizations of interval graphs, of
cocomparability graphs and of any ordering of the vertex set $V$ that can
result from an LDFS search of an arbitrary graph $G=(V,E)$. The following
ordering characterizes interval graphs and has appeared in a number of
papers including~\cite{Olariu91}.

\begin{lemma}[\hspace{-0.0005cm}\protect\cite{Olariu91}]
\label{int-order} $G=(V,E)$ is an interval graph if and only if there is an
ordering (called an~\emph{I-ordering}) of $V$ such that for all $x < y < z$,
if $xz \in E$ then also $xy \in E$.
\end{lemma}

Note that the characterization of Lemma~\ref{int-order} can also result after sorting 
the intervals of an interval representation of $G$ according to their left endpoints. 
Furthermore, note that some papers on interval graphs 
(for instance~\cite{Arikati90,Damasc93,longest-int-algo}) 
used the equivalent ``reverse'' vertex ordering, which results after sorting 
the intervals of an interval representation according to their right endpoints. 

A similar characterization of \emph{unit interval} graphs 
(also known as \emph{proper interval} graphs) requires that if $xz\in E$ then both $xy,yz \in E$. 
It was observed in~\cite{KratschStewart93} that the following generalization of the interval order
characterization captures cocomparability graphs.

\begin{lemma}[\hspace{-0.0005cm}\protect\cite{KratschStewart93}]
\label{cocomp-order} $G=(V,E)$ is a cocomparability graph if and only if
there is an ordering (called an \emph{umbrella-free ordering}, or a \emph{CO-ordering}) 
of $V$ such that for all $x < y < z$, if~$xz \in E$ then $xy \in E$
or $yz \in E$ (or both).
\end{lemma}

\begin{observation}
\label{interval-umbrella-free-obs} An I-ordering of an interval graph $G$ is
also an umbrella-free ordering.
\end{observation}

In the following, we present the notion of the recently introduced 
Lexicographic Depth First Search (LDFS) ordering (see~\cite{Corneil-LDFS08}).

\begin{definition}
\label{good-def} Let $G=(V,E)$ be a graph and $\sigma$ be any ordering of $V$. 
Let $(a,b,c)$ be a triple of vertices of $G$ such that $a <_{\sigma} b
<_{\sigma} c$, $ac \in E$, and $ab \notin E$. If there exists a vertex $d\in
V$ such that $a <_{\sigma} d <_{\sigma} b$, $db \in E$, and $dc \notin E$,
then $(a,b,c)$ is a \emph{good triple}; otherwise it is a \emph{bad triple}. 
Furthermore, if the triple $(a,b,c)$ is good, then vertex $d$ is called a 
\emph{$d$-vertex} of this triple.
\end{definition}

\begin{definition}
\label{ldfs-order} Let $G=(V,E)$ be a graph. An ordering $\sigma$ of $V$ is
a \emph{Lexicographic Depth First Search (LDFS) ordering} if and only if $%
\sigma$ has no bad triple.
\end{definition}

An example of a good triple $(a,b,c)$ and a $d$-vertex of it is depicted
in Figure~\ref{good-triple-fig}. In this example, the edges $ac$ and $db$
are indicated with solid lines, while the non-edges $ab$ and $dc$ are
indicated with dashed lines. Furthermore, the $d$-vertex is drawn gray for
better visibility.

\begin{figure}[h]
\centering
\includegraphics[scale=0.9]{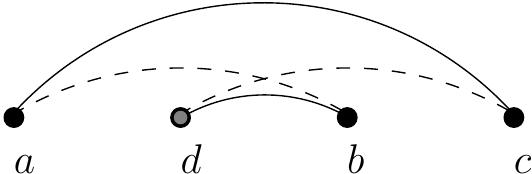}
\caption{A good triple $(a,b,c)$ and a $d$-vertex of this triple, in the
vertex ordering $\protect\sigma=(a,d,b,c)$.}
\label{good-triple-fig}
\end{figure}

\subsection{Algorithms}
\label{prelim-algorithms}

In the following we present
 the generic LDFS algorithm (Algorithm~\ref{ldfs-alg}) 
that starts at a distinguished vertex $u$. This algorithm has
been recently introduced in~\cite{Corneil-LDFS08}. It looks superficially
similar to the well known and well studied \emph{Lexicographic Breadth First
Search (LBFS)}~\cite{Rose-LBFS} (for a survey, see~\cite{Corneil-LBFS-survey}); 
nevertheless, it appears that vertex orderings computed by the LDFS and
by the LBFS have inherent structural differences. Briefly, the
generic LDFS algorithm proceeds as follows. Initially, the label $\varepsilon
$ is assigned to all vertices. Then, iteratively, an unvisited vertex~$v$
with lexicographically maximum label is chosen and removed from the graph.
If $v$ is chosen as the $i$th vertex, then each of its neighbors that are
still unnumbered have their label updated by having the digit $i$ \emph{prepended} 
to their label. The digits in the label of any vertex are always
in decreasing order, which ensures that all neighbors of the last chosen
vertex have a lexicographically greater label than its non-neighbors. By
extension, this argument ensures that all vertices are visited in a
depth-first search order. When applied to a graph with $n$ vertices and~$m$~edges, 
Algorithm~\ref{ldfs-alg} can be implemented to run in $O(\min\{n^{2}, n + m \log n\})$ 
time~\cite{Krueger05}; however, the current fastest implementation runs in 
$O(\min\{n^2, n+m \log\log n\})$~\cite{Spinrad-submitted}.

\begin{algorithm}[htb]
\caption{LDFS($G,u$)~\cite{Corneil-LDFS08}} \label{ldfs-alg}
\begin{algorithmic}[1]
\REQUIRE{A connected graph $G=(V,E)$ with $n$ vertices and a distinguished vertex $u$ of $G$}
\ENSURE{An LDFS ordering $\sigma_u$ of the vertices of $G$}

\medskip

\STATE{Assign the label $\varepsilon$ to all vertices}
\STATE{$label(u) \leftarrow \{0\}$}

\FOR{$i=1$ to $n$}
     \STATE{Pick an unnumbered vertex $v$ with the lexicographically largest label} \label{ldfs-alg-star}
     \STATE{$\sigma_u(i) \leftarrow v$} \COMMENT{assign to $v$ the number $i$}
     \FOR{each unnumbered vertex $w\in N(v)$}
          \STATE{prepend $i$ to $label(w)$}
     \ENDFOR
\ENDFOR

\medskip

\RETURN{the ordering $\sigma_{u} = (\sigma_{u}(1),\sigma_{u}(2),\ldots,\sigma_{u}(n))$} 
\end{algorithmic}
\end{algorithm}

The execution of the LDFS algorithm is captured in the example shown in
Figure~\ref{ldfs-fig}. In this example, suppose that the LDFS algorithm
starts at vertex $e$. Suppose that LDFS chooses vertex $d$ next. Now,
ordinary DFS could choose either $a$ or $c$ next, but LDFS has to choose $c$, 
since it has a greater label ($c$ was a neighbor of the previously visited
vertex $e$). The vertex following $c$ in the LDFS ordering $\sigma_{e}$ must
be $a$ rather than $b$, since $a$ has a greater label than $b$ ($a$~was a
neighbor of the previously visited vertex $d$). The LDFS then backtracks to $%
b$ completing the LDFS ordering as $\sigma_{e} = (e,d,c,a,b)$.

\begin{figure}[h]
\centering
\includegraphics[scale=0.9]{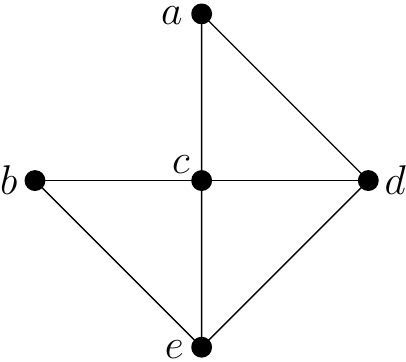}
\caption{Illustrating LDFS.}
\label{ldfs-fig}
\end{figure}

It is important here to connect the vertex ordering $\sigma_u$ that is
returned by the LDFS algorithm (i.e.~Algorithm~\ref{ldfs-alg}) with the
notion of an LDFS ordering, as defined in Definition~\ref{ldfs-order}. The
following theorem shows that a vertex ordering $\sigma$ of an arbitrary
graph $G$ can be returned by an application of the LDFS algorithm to $G$
(starting at some vertex $u$ of $G$) if and only if $\sigma$ is an LDFS
ordering.

\vspace{-0.1cm}
\begin{theorem}[\hspace{-0.0005cm}\protect\cite{Corneil-LDFS08}]
For an arbitrary graph $G=(V,E)$, an ordering $\sigma$ of $V$ can be
returned by an application of Algorithm~\ref{ldfs-alg} to $G$ if and only if 
$\sigma$ is an LDFS ordering.
\end{theorem}
\vspace{-0.1cm}

In the generic LDFS, there could be some choices to be made at line~\ref%
{ldfs-alg-star} of Algorithm~\ref{ldfs-alg}; in particular, at some
iteration there may be a set $S$ of vertices that have the same label and
the algorithm must choose one vertex from $S$. 
Generic LDFS (i.e.~Algorithm~\ref{ldfs-alg}) allows an arbitrary choice here. 
We present in the following a special kind of an LDFS algorithm, called
LDFS$^{+}$ (cf.~Algorithm~\ref{ldfs+-alg} below), which makes a specific
choice of a vertex in such a case of equal labels, as follows. Along with
the graph $G=(V,E)$, an ordering $\pi$ of $V$ is also given as input. 
The algorithm LDFS$^{+}$ (see Algorithm~\ref{ldfs+-alg} for a formal description) 
operates exactly as a generic LDFS that starts at the \emph{rightmost} vertex 
of $V$ in the ordering $\pi$, with the only difference that, in the case where 
at some iteration at least two unvisited vertices have the same label, it chooses the 
\emph{rightmost} vertex among them in the ordering $\pi$.

\begin{algorithm}[thb]
\caption{LDFS$^+$ ($G,\pi$)} \label{ldfs+-alg}
\begin{algorithmic}[1]
\REQUIRE{A connected graph $G=(V,E)$ with $n$ vertices and an ordering $\pi$ of $V$}
\ENSURE{An LDFS ordering $\sigma$ of the vertices of $G$}

\medskip

\STATE{Assign the label $\varepsilon$ to all vertices}

\FOR{$i=1$ to $n$}
     \STATE{Pick the rightmost vertex $v$ in $\pi$ among the unnumbered vertices with the  lexicographically largest label} \label{ldfs+-alg-star}
     \STATE{$\sigma(i) \leftarrow v$} \COMMENT{assign to $v$ the number $i$}
     \FOR{each unnumbered vertex $w\in N(v)$}
          \STATE{prepend $i$ to $label(w)$}
     \ENDFOR
\ENDFOR

\medskip

\RETURN{the ordering $\sigma = (\sigma(1),\sigma(2),\ldots,\sigma(n))$}
\end{algorithmic}
\end{algorithm}

In the following, we present the \emph{RightMost-Neighbor (RMN)} algorithm.
Although the name RMN has not been used, this algorithm essentially has been introduced in~\cite{Arikati90}
in order to find a minimum path cover in a given interval graph\footnote{Actually, 
the algorithm of~\cite{Arikati90} uses the ``reverse'' vertex 
ordering of an I-ordering
(as defined in Lemma~\ref{int-order}), which results after sorting the intervals of an interval
representation according to their right endpoints, and thus they presented
an equivalent \emph{LeftMost-Neighbor (LMN)} algorithm for the case of interval graphs. 
A similar observation applies to the algorithms in~\cite{Damasc93} and~\cite{longest-int-algo}.\vspace{-0.51cm}}.
The RMN algorithm is a very simple ``greedy'' algorithm that starts at the
rightmost vertex of a given ordering $\sigma$ of $V$ and traces each path
by repeatedly proceeding to the rightmost unvisited neighbor of the current vertex.
If the current vertex has no unvisited neighbors, then the rightmost unvisited vertex
is chosen as the first vertex in the next path.


\begin{algorithm}[htb]
\caption{RMN($\sigma$)} \label{rmn-alg}
\begin{algorithmic}[1]
\REQUIRE{A graph $G=(V,E)$ with $n$ vertices and an ordering $\sigma$ of $V$}
\ENSURE{An ordering $\widehat{\sigma}$ of the vertices of $G$}

\medskip

\STATE{Label all vertices as ``unvisited''; $i \leftarrow 1$}

\WHILE{there are unvisited vertices}
     \STATE{Pick the rightmost unvisited vertex $x$ in $\sigma$}
     \STATE{$\widehat{\sigma}(i) \leftarrow x$} \COMMENT{add vertex $x$ to the ordering $\widehat{\sigma}$}
     \STATE{Mark $x$ as ``visited''; $i \leftarrow i+1$}
     
     \WHILE{$x$ has at least one unvisited neighbor}
          \STATE{Pick the $x$'s rightmost unvisited neighbor $y$ in $\sigma$}
          \STATE{$\widehat{\sigma}(i) \leftarrow y$} \COMMENT{add vertex $y$ to the ordering $\widehat{\sigma}$}
          \STATE{Mark $y$ as ``visited''; $i \leftarrow i+1$}
          \STATE{$x \leftarrow y$}
     \ENDWHILE
     
\ENDWHILE

\medskip

\RETURN{$\widehat{\sigma}=(\widehat{\sigma}(1),\widehat{\sigma}(2),\ldots,\widehat{\sigma}(n))$} 
\end{algorithmic}
\end{algorithm}

Note that in Algorithm~\ref{ldfs+-alg} we denote the input vertex ordering by $\pi$
and the output ordering by~$\sigma$, 
while in Algorithm~\ref{rmn-alg}, $\sigma$ denotes the input vertex ordering. 
The reason for this notation is that we will often 
consider an arbitrary umbrella-free vertex ordering $\pi$ of a
cocomparability graph $G$, apply Algorithm~\ref{ldfs+-alg} (i.e.~LDFS$^{+}$)
to $\pi$ to compute the ordering $\sigma$, and then apply Algorithm~\ref{rmn-alg} 
(i.e.~RMN) to $\sigma$ to compute the ordering $\widehat{\sigma}$.
Then, as proved in~\cite{Corneil-MPC}, 
the LDFS vertex ordering $\sigma$ 
remains umbrella-free. 
Moreover, the final ordering~$\widehat{\sigma}$ defines a minimum path cover of $G$~\cite{Corneil-MPC}.

\section{Normal paths in cocomparability graphs\label{sec:normal}}

In this section we investigate the structure of the vertex ordering $\sigma$
that is obtained after applying an LDFS$^{+}$ preprocessing step to an
arbitrary umbrella-free ordering $\pi$ of a cocomparability graph $G$. On
such an LDFS umbrella-free ordering $\sigma$, we define a special type of
paths, called \emph{normal} paths (cf.~Definition~\ref{normal-def}), which
is a crucial notion for our algorithm for the longest path problem on
cocomparability graphs (cf.~Algorithm~\ref{alg-lp}). 
In the following definition we introduce the notion of a \emph{maximal} path 
in a graph, which extends that of a longest path.

\begin{definition}
\label{maximal-path}A path $P$ of a graph $G$ is \emph{maximal} if there
exists no path $P^{\prime }$ of $G$, such that $V(P)\subset V(P^{\prime })$.
\end{definition}

The main result of this section is that for any maximal path $P$ of a
cocomparability graph~$G$ (and thus also for any longest path), there exists
a normal path $P^{\prime}$ on the same vertices (cf.~Theorem~\ref{normal-thm}). 
Due to this result, it is sufficient for our algorithm that computes the 
longest path problem on cocomparability graphs (cf.~Algorithm~\ref{alg-lp}) 
to search only among the normal paths of the given cocomparability graph, 
in order to compute a longest path. The next lemma will be used in the sequel.

\begin{lemma}
\label{auxil}Let $G=(V,E)$ be a cocomparability graph and $\sigma $ be an
LDFS umbrella-free ordering of $V$. Let ${P=(v_{1},v_{2},\ldots ,v_{k})}$ be
a path of $G$ and $v_{\ell }\notin V(P)$ be a vertex of $G$, such that $%
v_{k}<_{\sigma }v_{\ell }<_{\sigma }v_{1}$ and $v_{\ell }v_{k}\notin E$.
Then, there exist two consecutive vertices $v_{i-1}$ and $v_{i}$ in $P$, $%
2\leq i\leq k$, such that $v_{i-1}v_{\ell }\in E$ and $v_{i}<_{\sigma
}v_{\ell }$.
\end{lemma}

\begin{proof}
Since $v_{k}<_{\sigma }v_{\ell }<_{\sigma }v_{1}$ and $v_{\ell }\notin V(P)$, there exists at least one
edge $e=xy$ of $P$, which straddles $v_{\ell }$ in~$\sigma $. Thus, at least
one of $x$ and $y$ is adjacent to $v_{\ell }$, since $\sigma $ is
umbrella-free. Recall that $v_{\ell }v_{k}\notin E$; let $v_{i-1}$, $2\leq
i\leq k$, be the last vertex of $P$ such that $v_{i-1}v_{\ell }\in E$. 
If~$v_{\ell }<_{\sigma }v_{i}$, then there exists similarly at least one vertex $%
v_{j}$, $i\leq j\leq k$, such that $v_{j}v_{\ell }\in E$, which is a
contradiction by the assumption on $v_{i-1}$. Thus, $v_{i}<_{\sigma }v_{\ell}$. 
This completes the proof of the lemma.
\end{proof}

\begin{definition}
\label{LDFS-closure}
Let $G=(V,E)$ be a cocomparability graph, $\sigma$ be an LDFS umbrella-free 
ordering of $V$, and $\sigma^{\prime}$ be an induced subordering of~$\sigma$. 
An \emph{LDFS closure} $\sigma^{\prime\prime}$ of~$\sigma^{\prime}$ 
(within $\sigma$) is an induced subordering of~$\sigma$ with the smallest 
number of vertices, such that $\sigma^{\prime\prime}$ is an LDFS ordering 
that includes $\sigma^{\prime}$.
\end{definition}

Observe that any induced subordering~$\sigma^{\prime}$ 
of an umbrella-free ordering~$\sigma$ also remains an umbrella-free ordering (cf.~Lemma~\ref{cocomp-order}). 
An example of a cocomparability graph $G=(V,E)$, as well as 
an LDFS umbrella-free ordering $\sigma=(u_{1},u_{2},\ldots,u_{9})$ of $V$ is 
illustrated in Figure~\ref{sigma-fig}. 
In this example, $\sigma^{\prime}=(u_{1},u_{3},u_{4},u_{5},u_{7},u_{8})$ is an 
induced subordering of $\sigma$ (and thus also umbrella-free). 
Furthermore, the ordering $\sigma^{\prime\prime}=(u_{1},u_{2},u_{3},u_{4},u_{5},u_{6},u_{7},u_{8})$ 
is an LDFS closure of~$\sigma^{\prime}$ (within~$\sigma$), where the vertices 
$u_{2}$ and $u_{6}$ are the $d$-vertices of the triples $(u_{1},u_{3},u_{4})$ 
and~$(u_{5},u_{7},u_{8})$, respectively.

\begin{figure}[h!]
\centering
\subfigure[]{ \label{sigma-graph-fig}
\includegraphics[scale=0.8]{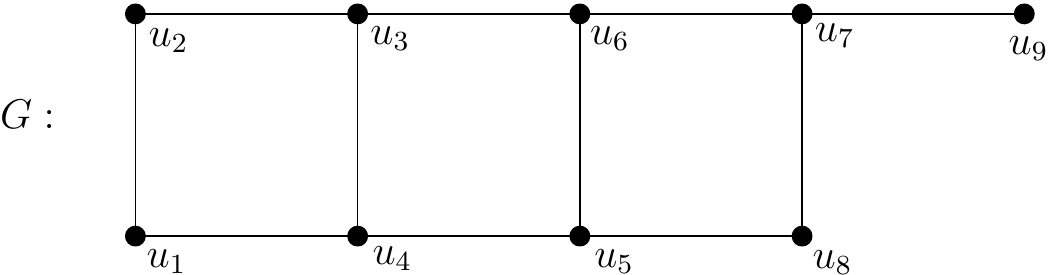}}
\subfigure[]{ \label{sigma-ordering-fig}
\includegraphics[scale=0.9]{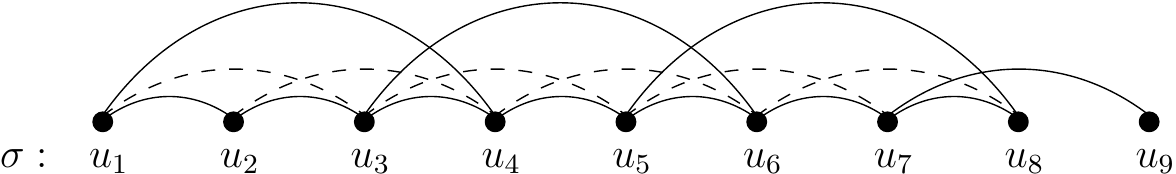}}
\caption{(a) A cocomparability graph $G=(V,E)$ and (b) an LDFS umbrella-free 
ordering~${\sigma=(u_{1},u_{2},\ldots,u_{9})}$ of $V$.\vspace{-0.2cm}}
\label{sigma-fig}
\end{figure}

\begin{observation}
\label{vertices-closure}Let $\sigma $ be an LDFS umbrella-free ordering, $%
\sigma ^{\prime }$ an arbitrary induced subordering of $\sigma $, and $%
\sigma ^{\prime \prime }$  any LDFS closure~of~$\sigma ^{\prime }$ (within 
$\sigma $). Then, every vertex $v$ of~$\sigma ^{\prime \prime }\setminus
\sigma ^{\prime }$ is a $d$-vertex of some good triple $(a,b,c)$ in $\sigma
^{\prime \prime }$.
\end{observation}

The next lemma follows easily by Observation~\ref{vertices-closure}.

\begin{lemma}
\label{infinite-sequence-closure}Let $\sigma $ be an LDFS umbrella-free
ordering, $\sigma ^{\prime }$ an arbitrary induced subordering of $\sigma $,
and $\sigma ^{\prime \prime }$  any LDFS closure~of~$\sigma ^{\prime }$
(within $\sigma $). Let $v$ be a vertex of~$\sigma ^{\prime \prime
}\setminus \sigma ^{\prime }$. Then, there exists at least one vertex $%
v^{\prime }$ in $\sigma ^{\prime }$, such that $v<_{\sigma }v^{\prime }$ and 
$vv^{\prime }\notin E$.
\end{lemma}

\begin{proof}
Suppose otherwise that for every vertex $v^{\prime }$ of $\sigma ^{\prime
\prime }$, for which $v<_{\sigma }v^{\prime }$ and $vv^{\prime }\notin E$,
the vertex $v^{\prime }$ belongs to $\sigma ^{\prime \prime }\setminus
\sigma ^{\prime }$. Note by Observation~\ref{vertices-closure} that $v$ is a 
$d$-vertex of some good triple in $\sigma ^{\prime \prime }$; let this
triple be $(a_{0},b_{0},c_{0})$. Then, $v<_{\sigma }c_{0}$ and $vc_{0}\notin
E$ by definition of a good triple. Thus, $c_{0}$ is a vertex of $\sigma
^{\prime \prime }\setminus \sigma ^{\prime }$ by our assumption on $v$.
Then, $c_{0}$ is a $d$-vertex of some good triple in $\sigma ^{\prime \prime
}$ by Observation~\ref{vertices-closure}; let this triple be $%
(a_{1},b_{1},c_{1})$. Then, in particular, $c_{0}<_{\sigma }c_{1}$ and $%
c_{0}c_{1}\notin E$ by definition of a good triple. Therefore $v<_{\sigma
}c_{0}<_{\sigma }c_{1}$. Furthermore $vc_{1}\notin E$, since otherwise the
vertices $v,c_{0},c_{1}$ build an umbrella in $\sigma $, which is a
contradiction. That is, $v<_{\sigma }c_{1}$ and $vc_{1}\notin E$, and thus $%
c_{1}$ is a vertex of $\sigma ^{\prime \prime }\setminus \sigma ^{\prime }$
by our assumption on $v$. Now, for every $i\geq 2$, we can inductively
construct a sequence $c_{2},c_{3},\ldots ,c_{i}$ of vertices in $\sigma
^{\prime \prime }\setminus \sigma ^{\prime }$, such that $v<_{\sigma
}c_{0}<_{\sigma }c_{1}<_{\sigma }c_{2}<_{\sigma }\ldots <_{\sigma }c_{i}$
and the vertices $v,c_{0},c_{1},c_{2},\ldots ,c_{i}$ build an independent
set. This is a contradiction, since $\sigma ^{\prime \prime }$ is finite.
Therefore, there exists at least one vertex $v^{\prime }$ in~$\sigma
^{\prime }$, such that $v<_{\sigma }v^{\prime }$ and $vv^{\prime }\notin E$.
This completes the proof of the lemma.
\end{proof}

\begin{corollary}
\label{rightmost-in-closure}Let $\sigma $ be an LDFS umbrella-free ordering, 
$\sigma ^{\prime }$ an arbitrary induced subordering of~$\sigma $, and $%
\sigma ^{\prime \prime }$ any LDFS closure~of~$\sigma ^{\prime }$ (within 
$\sigma $). Then, the rightmost vertex of $\sigma ^{\prime }$ is also the
rightmost vertex of $\sigma ^{\prime \prime }$.
\end{corollary}

\begin{proof}
Let $v^{\prime }$ and $v^{\prime \prime }$ be the rightmost vertices of $%
\sigma ^{\prime }$ and of $\sigma ^{\prime \prime }$, respectively. If $%
v^{\prime }\neq v^{\prime \prime }$, then $v^{\prime \prime }$ is a vertex
of $\sigma ^{\prime \prime }\setminus \sigma ^{\prime }$ and $v^{\prime
}<_{\sigma }v^{\prime \prime }$, since $\sigma ^{\prime }$ is a subset of $%
\sigma ^{\prime \prime }$. Then, there exists by Lemma~\ref%
{infinite-sequence-closure} at least one vertex $v^{\prime \prime \prime }$
in $\sigma ^{\prime }$, such that $v^{\prime \prime }<_{\sigma }v^{\prime
\prime \prime }$, i.e.~$v^{\prime }<_{\sigma }v^{\prime \prime }<_{\sigma
}v^{\prime \prime \prime }$, which is a contradiction to our assumption on $%
v^{\prime }$.
\end{proof}

\medskip

In the following we introduce the notion of a typical and a normal path in a 
cocomparability graph $G=(V,E)$ (with respect to an LDFS umbrella-free 
ordering $\sigma$ of $V$), which will be used in the remainder of the paper.

\begin{definition}
\label{normal-def}Let $G=(V,E)$ be a cocomparability graph and $\sigma $ be
an LDFS umbrella-free ordering of $V$. Then,\vspace{-0.1cm}
\begin{enumerate}
\item[(a)] a path ${P=(v_{1},v_{2},\ldots ,v_{k})}$ of $G$ is called \emph{%
typical} if $v_{1}$ is the rightmost vertex of $V(P)$ in~$\sigma $ and $%
v_{2} $ is the rightmost vertex of $N(v_{1})\cap V(P)$ in~$\sigma $, and\vspace{-0.1cm}
\item[(b)] a typical path ${P=(v_{1},v_{2},\ldots ,v_{k})}$ of $G$ is called 
\emph{normal} if $v_{i}$ is the rightmost vertex of~$N(v_{i-1})\cap
\{v_{i},v_{i+1},\ldots ,v_{k}\}$ in~$\sigma $, for every $i=2,\ldots ,k$.\vspace{-0.15cm}
\end{enumerate}
\end{definition}

For example, in the cocomparability graph $G$ of Figure~\ref{sigma-fig}, 
the path $P=(u_{8},u_{5},u_{6},u_{3},u_{4})$ is a normal path. 
The next observation follows from Definition~\ref{normal-def}.

\begin{observation}
\label{normal-RMN}Let $G=(V,E)$ be a cocomparability graph and $\sigma $ be
an LDFS umbrella-free ordering of $V$. Let $P$ be a normal path of $G$ (with
respect to the ordering $\sigma $) and $\sigma |_{V(P)}$ be the restriction
of $\sigma $ on the vertices of $P$. Then, the ordering of the vertices of $%
V(P)$ in $P$ coincides with the ordering RMN$(\sigma |_{V(P)})$.
\end{observation}

A similar notion of a normal (i.e.~RMN) path for the special case of interval graphs 
has appeared in~\cite{Damasc93} (referred to as a \emph{straight} path), 
as well as in~\cite{longest-int-algo}. 
We now state the following two auxiliary lemmas.

\begin{lemma}[\hspace{-0.0005cm}\protect\cite{Corneil-MPC}]
\label{auxil-2}Let $G=(V,E)$ be a cocomparability graph and $\pi $ be an
umbrella-free ordering of $V$. Let $\pi ^{\prime }=$RMN$(\pi )$ and $\pi
^{\prime \prime }=$LDFS$^{+}(\pi )$. Furthermore, let $x,y\in V$ such that $%
xy\notin E$. If~$y<_{\pi }x$, then $x<_{\pi ^{\prime }}y$ and $x<_{\pi
^{\prime \prime }}y$.
\end{lemma}

\begin{lemma}
\label{auxil-3}Let $G=(V,E)$ be a cocomparability graph, $\pi $ be an
umbrella-free ordering of $V$, and~$\pi ^{\prime }=$RMN$(\pi )$. Let $x,y\in
V$ such that $y<_{\pi }x$ and $y<_{\pi ^{\prime }}x$. Then, $y$ is not the
first vertex of $\pi ^{\prime }$ and for the previous vertex $z$ of $y$ in $%
\pi ^{\prime }$, $y<_{\pi }x<_{\pi }z$, $zy\in E$, and $zx\notin E$.
\end{lemma}

\begin{proof}
First, note that the first vertex of the ordering $\pi ^{\prime }=$RMN$(\pi
) $ is the rightmost vertex of~$\pi $. Thus $y$ is not the first vertex of $%
\pi ^{\prime }$, since $y<_{\pi }x$. Let $z$ be the previous vertex of $y$
in $\pi ^{\prime }$. Then, $x$ is unvisited, when $z$ is being visited by $%
\pi ^{\prime }$, since $z<_{\pi ^{\prime }}y<_{\pi ^{\prime }}x$. Suppose
that $zx\in E$. Then, $y$ could not be the next vertex of $z$ in $\pi
^{\prime }$, since $x$ is unvisited and $y<_{\pi }x$%
, which is a contradiction. Thus $zx\notin E$. Furthermore, Lemma~\ref%
{auxil-2} implies that $x<_{\pi }z$, i.e.~$y<_{\pi }x<_{\pi }z$, since $%
z<_{\pi ^{\prime }}x$ and $zx\notin E$. Suppose now that $zy\notin E$. In
the case where no neighbor of $z$ is unvisited, when $z$ is being visited by 
$\pi ^{\prime }$, then $x$ is the next vertex of $z$ in $\pi ^{\prime }$
instead of $y$, since $y<_{\pi }x$, which is a contradiction. In the case
where at least one neighbor $w$ of $z$ is unvisited, when $z$ is being
visited by $\pi ^{\prime }$, then one of the unvisited neighbors $w$ of $z$
is the next vertex of~$z$ in $\pi ^{\prime }$ instead of $y$, which is again
a contradiction. Thus, $zy\in E$. This completes the proof of the lemma.
\end{proof}

\begin{notation}
\label{maximal-path-orderings-notation}
In the remainder of this section, we consider 
a cocomparability graph $G=(V,E)$ and 
an \emph{LDFS umbrella-free ordering} $\sigma$ of~$G$. 
Furthermore, we consider 
a \emph{maximal} path $P$ of~$G$, 
the \emph{restriction} $\sigma^{\prime }=\sigma |_{V(P)}$ of $\sigma$ on the vertices of $P$ and
an arbitrary \emph{LDFS~closure}~$\sigma^{\prime\prime}$ of~$\sigma^{\prime}$ (within~$\sigma$). 
Finally, we consider 
the orderings~$\widehat{\sigma }=$LDFS$^{+}(\sigma ^{\prime })$ 
and~$\widehat{\widehat{\sigma }}=$RMN$(\widehat{\sigma})$. 
\end{notation}

The next structural lemma will be used in the sequel, in order to prove in
Theorem~\ref{normal-thm} that for every maximal path $P$ there exists a 
normal path $P^{\prime}$ of~$G$, such that $V(P^{\prime})=V(P)$.

\begin{lemma}
\label{typical-order}Let $x,y,z$ be three vertices of $\sigma ^{\prime }$,
such that ${x<_{\widehat{\sigma }}y<_{\widehat{\sigma }}z}$ 
and~${z<_{\sigma^{\prime}}y<_{\sigma^{\prime}}x}$, where~${xy,xz\in E}$ and ${yz\notin E}$.
Then, $x$ is not the next vertex of $z$ in $\widehat{\widehat{\sigma }}$.
\end{lemma}

\begin{proof}
The proof will be done by contradiction. We will exploit the facts 
that~$P$ is a maximal path (cf.~Notation~\ref{maximal-path-orderings-notation}) 
and that, given a Hamiltonian cocomparability graph $H$ 
and an LDFS umbrella-free ordering $\pi$ of $H$, 
the ordering RMN$(\pi)$ gives a Hamiltonian path of $H$~\cite{Corneil-MPC}. 
Suppose that there exists a triple $(x,y,z)$ of vertices in $\sigma ^{\prime }$ 
that satisfy the conditions of the lemma, 
such that $x$ is the next vertex of $z$ in $\widehat{\widehat{\sigma }}$.
Among all those triples, let $(x,y,z)$ be the one, where $z$ is the
rightmost possible in $\widehat{\sigma }$ and $y$ is the rightmost possible
in $\widehat{\sigma }$ among those with equal $z$. Note that always $z<_{%
\widehat{\widehat{\sigma }}}y$ by Lemma~\ref{auxil-2}, since $\widehat{%
\widehat{\sigma }}=$RMN$(\widehat{\sigma })$, and since $yz\notin E$ and $%
y<_{\widehat{\sigma }}z$ by assumption.

Since $z<_{\sigma ^{\prime }}y<_{\sigma ^{\prime }}x$, $xz\in E$, and $%
yz\notin E$, and since $\sigma ^{\prime \prime }$ is an LDFS-closure of $%
\sigma ^{\prime }$ (within $\sigma $), there exists by Observation~\ref%
{vertices-closure} a vertex $d$ in $\sigma ^{\prime \prime }$, such that $%
z<_{\sigma ^{\prime \prime }}d<_{\sigma ^{\prime \prime }}y<_{\sigma
^{\prime \prime }}x$, $dy\in E$, and $dx\notin E$. Thus, since $dx\notin E$
and $\sigma $ is an umbrella-free ordering, it follows that $zd\in E$. 
In the following we
will distinguish  the cases where $d\in \sigma ^{\prime }$
and $d\in \sigma ^{\prime \prime }\setminus \sigma ^{\prime }$.

\medskip

\emph{Case 1.} Suppose first that $d\in \sigma ^{\prime }$, i.e.~$d\in V(P)$, 
and thus $d\in \widehat{\sigma }$. Then, since $\widehat{\sigma }=$LDFS$%
^{+}(\sigma ^{\prime })$, and since $dx\notin E$ and $d<_{\sigma ^{\prime
}}x $, it follows by Lemma~\ref{auxil-2} that $x<_{\widehat{\sigma }}d$.

\medskip

\emph{Case 1a.} Suppose that $d<_{\widehat{\sigma }}z$, i.e.~$x<_{\widehat{%
\sigma }}d<_{\widehat{\sigma }}z$. If $d$ is unvisited when $z$ is being
visited in~$\widehat{\widehat{\sigma }}$, then $d$ would be the next vertex
of $z$ in $\widehat{\widehat{\sigma }}$ instead of $x$, since $zd\in E$,
which is a contradiction. Thus, $d$ has been visited before $z$ in $\widehat{%
\widehat{\sigma }}$, i.e.~$d<_{\widehat{\widehat{\sigma }}}z$. Therefore,
Lemma~\ref{auxil-3} implies that $d$ is not the first vertex in $\widehat{%
\widehat{\sigma }}$, while $d<_{\widehat{\sigma }}z<_{\widehat{\sigma }}a$, $%
ad\in E$, and $az\notin E$ for the previous vertex $a$ of $d$ in $\widehat{%
\widehat{\sigma }}$. Then, in particular, $a<_{\sigma ^{\prime }}z$ by 
Lemma~\ref{auxil-2}, since $z<_{\widehat{\sigma }}a$, $az\notin E$, and $\widehat{%
\sigma }=$LDFS$^{+}(\sigma ^{\prime })$. Summarizing, $d<_{\widehat{\sigma }%
}z<_{\widehat{\sigma }}a$ and $a<_{\sigma ^{\prime }}z<_{\sigma ^{\prime }}d$, 
where $dz,da\in E$ and $za\notin E$, while $d$ is the next vertex of $a$
in $\widehat{\widehat{\sigma }}$. This comes in contradiction to the choice
of the triple $(x,y,z)$, for which $z$ is the rightmost possible in $%
\widehat{\sigma }$.

\medskip

\emph{Case 1b.} Suppose that $z<_{\widehat{\sigma }}d$, i.e.~$y<_{\widehat{%
\sigma }}z<_{\widehat{\sigma }}d$. Recall that $yd\in E$ and $yz\notin E$.
Thus, since~$\widehat{\sigma }$ is an LDFS ordering, there exists a vertex $%
d^{\prime }$ in $\widehat{\sigma }$, such that $y<_{\widehat{\sigma }%
}d^{\prime }<_{\widehat{\sigma }}z<_{\widehat{\sigma }}d$, $d^{\prime }z\in
E $, and $d^{\prime }d\notin E$. Note that $x<_{\widehat{\sigma }}y<_{%
\widehat{\sigma }}d^{\prime }$. Similarly to the previous paragraph, if $%
d^{\prime }$ is unvisited when $z$ is being visited in $\widehat{\widehat{%
\sigma }}$, then $d^{\prime }$ would be the next vertex of $z$ in $\widehat{%
\widehat{\sigma }}$ instead of $x$, since $d^{\prime }z\in E$, which is a
contradiction. Thus, $d^{\prime }$ has been visited before $z$ in $\widehat{%
\widehat{\sigma }}$, i.e.~$d^{\prime }<_{\widehat{\widehat{\sigma }}}z$.
Therefore, Lemma~\ref{auxil-3} implies that $d^{\prime }$ is not the first
vertex in $\widehat{\widehat{\sigma }}$, while $d^{\prime }<_{\widehat{%
\sigma }}z<_{\widehat{\sigma }}a^{\prime }$, $a^{\prime }d^{\prime }\in E$,
and $a^{\prime }z\notin E$ for the previous vertex $a^{\prime }$ of $%
d^{\prime }$ in $\widehat{\widehat{\sigma }}$. Then, in particular, $%
a^{\prime }<_{\sigma ^{\prime }}z$ by Lemma~\ref{auxil-2}, since $z<_{%
\widehat{\sigma }}a^{\prime }$, $a^{\prime }z\notin E$, and $\widehat{\sigma 
}=$LDFS$^{+}(\sigma ^{\prime })$. Similarly, $d<_{\sigma ^{\prime
}}d^{\prime }$, since $d^{\prime }<_{\widehat{\sigma }}d$ and $d^{\prime
}d\notin E$. Therefore, since $z<_{\sigma ^{\prime }}d$, it follows that $%
z<_{\sigma ^{\prime }}d<_{\sigma ^{\prime }}d^{\prime }$. Summarizing, $%
d^{\prime }<_{\widehat{\sigma }}z<_{\widehat{\sigma }}a^{\prime }$ and $%
a^{\prime }<_{\sigma ^{\prime }}z<_{\sigma ^{\prime }}d^{\prime }$, where $%
d^{\prime }z,d^{\prime }a^{\prime }\in E$ and $za^{\prime }\notin E$, while $%
d^{\prime }$ is the next vertex of $a^{\prime }$ in $\widehat{\widehat{%
\sigma }}$. This comes in contradiction to the choice of the triple $(x,y,z)$, 
such that $z$ is the rightmost possible in $\widehat{\sigma }$.

\medskip

\emph{Case 2.} Suppose now that $d\in \sigma ^{\prime \prime }\setminus
\sigma ^{\prime }$, i.e.~$d\notin V(P)$, and thus $d\notin \widehat{\sigma }$. 
Consider the set of vertices $w$ of $\widehat{\sigma }$, such that $y<_{%
\widehat{\sigma }}w<_{\widehat{\sigma }}z$. We partition this set into the
(possibly empty) sets $A=\{w \ | \ y<_{\widehat{\sigma }}w<_{\widehat{\sigma }%
}z,wz\notin E\}$ and $B=\{w \ | \ y<_{\widehat{\sigma }}w<_{\widehat{\sigma }%
}z,wz\in E\}$. First observe that~$xw\in E$ for every $w\in A$, since $xz\in
E$ and $zw\notin E$, and since $\widehat{\sigma }$ is an umbrella-free
ordering. We will now prove that $yw\in E$ for every vertex $w\in A$.
Suppose otherwise that $yw,zw\notin E$ for a vertex $w$, for which $y<_{%
\widehat{\sigma }}w<_{\widehat{\sigma }}z$. Then, $z<_{\sigma ^{\prime
}}w<_{\sigma ^{\prime }}y$ by Lemma~\ref{auxil-2}, since $y<_{\widehat{%
\sigma }}w<_{\widehat{\sigma }}z$, $yw,zw\notin E$, and $\widehat{\sigma }=$%
LDFS$^{+}(\sigma ^{\prime })$. Recall that $y<_{\sigma ^{\prime }}x$ by
assumption in the statement of the lemma. Thus, $x<_{\widehat{\sigma }}w<_{%
\widehat{\sigma }}z$ and $z<_{\sigma ^{\prime }}w<_{\sigma ^{\prime }}x$,
where $xw,xz\in E$ and $wz\notin E$, while $x$ is the next vertex of $z$ in $%
\widehat{\widehat{\sigma }}$. Therefore, since $y<_{\widehat{\sigma }}w$,
this comes in contradiction to the choice of the triple $(x,y,z)$, such that 
$y$ is the rightmost possible (with respect to $z$) in $\widehat{\sigma }$.
Therefore, $yw\in E$ for every vertex $w\in A$.

If a vertex $w\in B$ is unvisited when $z$ is being visited in $\widehat{%
\widehat{\sigma }}$, then $w$ would be the next vertex of $z$ in $\widehat{%
\widehat{\sigma }}$ instead of $x$, which is a contradiction to the
assumption. Thus, $w$ has been visited before $z$ in $\widehat{\widehat{%
\sigma }}$, i.e.~$w<_{\widehat{\widehat{\sigma }}}z$, for every $w\in B$. On
the other hand, Lemma~\ref{auxil-2} implies that $z<_{\widehat{\widehat{%
\sigma }}}w$ for every $w\in A$, i.e.~$w$ is being visited after $z$ in $%
\widehat{\widehat{\sigma }}$, since $\widehat{\widehat{\sigma }}=$RMN$(%
\widehat{\sigma })$, $w<_{\widehat{\sigma }}z$, and $wz\notin E$ for every $%
w\in A$. Let $v$ be a vertex that is visited after $z$ in $\widehat{\widehat{%
\sigma }}$, i.e.~$z<_{\widehat{\widehat{\sigma }}}v$. Then, $v<_{\widehat{%
\sigma }}z$. Indeed, suppose otherwise that $z<_{\widehat{\sigma }}v$ for
such a vertex $v$. If $zv\in E$, then $v$ is the next vertex of $z$ in~$%
\widehat{\widehat{\sigma }}$ instead of $x$, since in this case $x<_{%
\widehat{\sigma }}z<_{\widehat{\sigma }}v$, which is a contradiction. If $%
zv\notin E$, then $v<_{\widehat{\widehat{\sigma }}}z$ by Lemma~\ref{auxil-2}, 
since $\widehat{\widehat{\sigma }}=$RMN$(\widehat{\sigma })$, which is
again a contradiction. Thus, $v<_{\widehat{\sigma }}z$ for every vertex $v$
that is visited after $z$ in $\widehat{\widehat{\sigma }}$. 
In the following we distinguish the cases $A\neq \emptyset $ and $A=\emptyset $. The
case where $A=\emptyset $ can be handled similarly to the case where $A\neq
\emptyset $, as we will see in the sequel.

\medskip

\emph{Case 2a.} $A\neq \emptyset $. Recall that $z<_{\widehat{\widehat{%
\sigma }}}w$ for every $w\in A$, i.e.~every $w\in A$ is visited after $z$ 
in~$\widehat{\widehat{\sigma }}$, as we proved above. Thus, since $x$ is the
next vertex of $z$ in $\widehat{\widehat{\sigma }}$ by assumption, all
vertices $w\in A$ are unvisited, when $x$ is being visited by~$\widehat{%
\widehat{\sigma }}$. Therefore, the next vertex of $x$ in $\widehat{\widehat{%
\sigma }}$ is some $w_{1}\in A$, since $xw\in E$ and $y<_{\widehat{\sigma }%
}w $ for every $w\in A$. Now recall by the previous paragraph that $v<_{%
\widehat{\sigma }}z$ for every vertex $v$ that is visited after $z$ in $%
\widehat{\widehat{\sigma }}$, and that all vertices of $B$ have been visited
before $z$ in $\widehat{\widehat{\sigma }}$. Furthermore recall that $yw\in
E $ for every $y\in A$. Therefore, $\widehat{\widehat{\sigma }}$ visits
after $w_{1}$ only vertices $w\in A$, until it reaches vertex $y$. Denote by 
$P^{\prime }$ the path on the vertices of $V(P)$ produced by $\widehat{%
\widehat{\sigma }}$. Suppose that not all vertices of $A$ have been visited
before $y$ in $P^{\prime }$, i.e.~in $\widehat{\widehat{\sigma }}$. Then the
next vertex of $y$ in $\widehat{\widehat{\sigma }}$ is again some $w_{2}\in
A $. That is, $P^{\prime }=(P_{0},z,x,P_{1},y,w_{2},P_{2})$ for some
subpaths $P_{0}$, $P_{1}$, and $P_{2}$ of $P^{\prime }$, where $%
V(P_{1})\subseteq A$ and $w_{2}\in A$. Thus, since $xw\in E$ for every $w\in
A$, there exists the path $P^{\prime \prime
}=(P_{0},z,d,y,P_{1},x,w_{2},P_{2})$, where $V(P^{\prime \prime })=V(P)\cup
\{d\}$, which is a contradiction, since $P$ is a maximal path.

Thus we may assume for the sequel that all vertices of $A$ have been visited
before $y$ in~$P^{\prime }$, i.e.~in $\widehat{\widehat{\sigma }}$. Then, $%
V(P_{1})=A$. If $y$ is the last vertex in $\widehat{\widehat{\sigma }}$,
then $P^{\prime }=(P_{0},z,x,P_{1},y)$ for some subpaths $P_{0}$ and $P_{1}$
of $P^{\prime }$, where $V(P_{1})=A$. In this case, there exists the 
path~$P^{\prime \prime }=(P_{0},z,d,y,P_{1},x)$, where $V(P^{\prime \prime
})=V(P)\cup \{d\}$, which is a contradiction, since~$P$ is a maximal path.
Suppose that $y$ is not the last vertex in $\widehat{\widehat{\sigma }}$ and
denote by $q\notin A$ the next vertex of $y$ in $\widehat{\widehat{\sigma }}$. 
Then, $P^{\prime }=(P_{0},z,x,P_{1},y,q,P_{2})$ for some subpaths $P_{0}$, 
$P_{1}$, and~$P_{2}$ of $P^{\prime }$, where $V(P_{1})=A$, and let $%
P_{1}=(w_{1},w_{2},\ldots ,w_{\ell })$. If $w_{\ell }q\in E$, there exists
the path~$P^{\prime \prime }=(P_{0},z,d,y,x,P_{1},q,P_{2})$, which
contradicts the maximality of $P$. If $xq\in E$, then there exists the 
path $P^{\prime \prime }=(P_{0},z,d,y,P_{1},x,q,P_{2})$, which again contradicts 
the maximality of~$P$.

To complete the proof of Case 2a, we now assume that $w_{\ell }q,xq\notin E$. 
First we proof that $q<_{\widehat{\sigma }}x$. Otherwise, suppose
that $y<_{\widehat{\sigma }}q$. Then $q<_{\widehat{\sigma }}z$, since $v<_{%
\widehat{\sigma }}z$ for every vertex $v$ that is visited after $z$ in $%
\widehat{\widehat{\sigma }}$, as we proved above, and thus $y<_{\widehat{%
\sigma }}q<_{\widehat{\sigma }}z$. Furthermore $q\notin B$, since all
vertices of $B$ have been visited before $z$ in $\widehat{\widehat{\sigma }}$%
, as we proved above. Therefore $y\in A$, which is a contradiction, since we
assumed that all vertices of $A$ have been visited before $y$ in $\widehat{%
\widehat{\sigma }}$. Now suppose  $x<_{\widehat{\sigma }}q<_{\widehat{\sigma }%
}y$, i.e.~$x<_{\widehat{\sigma }}q<_{\widehat{\sigma }}y<_{\widehat{\sigma }%
}w_{\ell }$. Then the vertices $x,q,w_{\ell }$ build an umbrella in $%
\widehat{\sigma }$, which is again a contradiction, since $\widehat{\sigma }$
is umbrella-free. Thus, $q<_{\widehat{\sigma }}x$.

Let $s$ be a vertex, such that $x<_{\widehat{\sigma }}s<_{\widehat{\sigma }%
}y $ and $s$ is visited after $y$ in $\widehat{\widehat{\sigma }}$. Then, $%
s\neq q$, since $q<_{\widehat{\sigma }}x<_{\widehat{\sigma }}s$. If $ys\in E$, 
then $s$ is the next vertex of $y$ in $\widehat{\widehat{\sigma }}$
instead of $q$, since $\widehat{\widehat{\sigma }}=$RMN$(\widehat{\sigma })$, 
which is a contradiction. Thus $ys\notin E$ for every vertex $s$, such
that $x<_{\widehat{\sigma }}s<_{\widehat{\sigma }}y$ and $s$ is visited
after $y$ in $\widehat{\widehat{\sigma }}$.

We now construct a new ordering $\rho $ of $V(P)\cup \{v\}$, where $v$ is a
new vertex. This new ordering $\rho $ is based on the LDFS umbrella-free
ordering $\widehat{\sigma }$ and the structure of $\rho $ will allow us to
show that $G$ has a path on the vertices of $V(P)\cup \{d\}$, thereby
contradicting the maximality of path $P$. The ordering $\rho $ is
constructed by adding the new vertex $v$ immediately to the right of vertex $%
y$ in $\widehat{\sigma }$. The adjacencies between the vertices of $V(P)$ in 
$\widehat{\sigma }$ remain the same in $\rho $, while the adjacencies
between the new vertex $v$ and the vertices of $V(P)$ in $\rho $ are defined
as follows. First, $v$ is made adjacent in $\rho $ to $y$ and to all
neighbors of $y$. Second, $v$ is made adjacent also to $z$ and to every
vertex $w\in B$. Note that $v$ is adjacent in $\rho $ to all vertices $w$ of 
$\widehat{\sigma }$, for which $v<_{\rho }w\leq _{\rho }z$. Therefore, if $%
wy\notin E$ and $wv\in E$ in $\rho $ for some vertex $w\in V(P)\setminus
\{y\}$, then $y<_{\widehat{\sigma }}w\leq _{\widehat{\sigma }}z$ (in
particular, $w\in B$). Let $H$ be the graph induced by the ordering $\rho $.

We will prove that $\rho $ remains an LDFS umbrella free ordering of the
vertices of $V(P)\cup \{v\}$. Since $G[V(P)]$ (i.e.~the subgraph of $G$
induced by $\widehat{\sigma }$) is an induced subgraph of $H$, if there is
an umbrella or a bad triple in $\rho $, then the new vertex $v$ must belong
to this umbrella or bad triple, since $\widehat{\sigma }=\rho |_{V(P)}$ is
an LDFS umbrella-free ordering of $V(P)$. Suppose that $v$ belongs to an
umbrella in $\rho $ with vertices $a,b,v$, where either $v<_{\rho }a<_{\rho
}b$, or $a<_{\rho }v<_{\rho }b$, or $a<_{\rho }b<_{\rho }v$.

Suppose first that $v<_{\rho }a<_{\rho }b$. Then, since $va\notin E$, it
follows by the construction of $\rho $ that $z<_{\rho }a$, i.e.~$z<_{\rho
}a<_{\rho }b$, and thus also $ya\notin E$ and $yb\in E$. That is, the
vertices $y,a,b$ build an umbrella in $\widehat{\sigma }$, which is a
contradiction. Suppose now that $a<_{\rho }v<_{\rho }b$. Then, $a\neq y$,
since $v$ is adjacent to $y$ in $\rho $. Thus, since $va\notin E$, it
follows by the construction of $\rho $ that also $ay\notin E$. Furthermore,
since $vb\notin E$, it follows by the construction of $\rho $ that $z<_{\rho
}b$, and thus also $yb\notin E$. That is, the vertices $a,y,b$ build an
umbrella in $\widehat{\sigma }$, which is a contradiction. Suppose finally
that $a<_{\rho }b<_{\rho }v$. Then, $b\neq y$, since $v$ is adjacent to $y$
in $\rho $. Furthermore, $a\neq y$, since $y$ lies immediately to the left
of $v$ in $\rho $. Thus, since $av\in E$ and $bv\notin E$, it follows by the
construction of $\rho $ that also $ay\in E$ and $by\notin E$, i.e.~the
vertices $a,b,y$ build an umbrella in $\widehat{\sigma }$, which is a
contradiction. Thus, $\rho $ is umbrella-free.

Suppose now that $v$ belongs to a bad triple in $\rho $ with vertices $a,b,v$, 
where either $v<_{\rho }a<_{\rho }b$ or $a<_{\rho }v<_{\rho }b$ or $%
a<_{\rho }b<_{\rho }v$. First let $v<_{\rho }a<_{\rho }b$, where $vb\in E$
and $va\notin E$. Since $va\notin E$, it follows by the construction of $%
\rho $ that $z<_{\widehat{\sigma }}a<_{\widehat{\sigma }}b$, and thus also $%
y<_{\widehat{\sigma }}a<_{\widehat{\sigma }}b$, $yb\in E$, and $ya\notin E$.
Since $\widehat{\sigma }$ is an LDFS ordering, there exists a vertex $%
v^{\prime }$ between $y$ and $a$ in $\widehat{\sigma }$, such that $%
v^{\prime }a\in E$ and $v^{\prime }b\notin E$. Note that $v^{\prime }\neq v$, 
since $vb\in E$ and $v^{\prime }b\notin E$. Thus, the vertices $v,a,b$ do
not build a bad triple in $\rho $, which is a contradiction. Now let $%
a<_{\rho }v<_{\rho }b$, where $ab\in E$ and $av\notin E$. Note that $a\neq y$, 
since $av\notin E$, and thus also $a<_{\widehat{\sigma }}y<_{\widehat{%
\sigma }}b$ and $ay\notin E$. Since~$\widehat{\sigma }$ is an LDFS ordering,
there exists a vertex $v^{\prime }$ between $a$ and $y$ in $\widehat{\sigma }
$, such that $v^{\prime }b\notin E$ and $v^{\prime }y\in E$, and thus also $%
v^{\prime }v\in E$. Thus, the vertices $a,v,b$ do not build a bad triple in $%
\rho $, which is a contradiction. Finally let $a<_{\rho }b<_{\rho }v$, where 
$av\in E$ and $ab\notin E$. By the construction of~$\rho $, note that $b\neq
y$, since $av\in E$ and $ab\notin E$. Thus $a<_{\widehat{\sigma }}b<_{%
\widehat{\sigma }}y$. Furthermore, $ay\in E$ by the construction of $\rho $,
since $av\in E$. Since $\widehat{\sigma }$ is an LDFS ordering, there exists
a vertex $v^{\prime }$ between~$a$ and $b$ in $\widehat{\sigma }$, such that 
$v^{\prime }b\in E$ and $v^{\prime }y\notin E$, and thus also $v^{\prime
}v\notin E$. Thus, the vertices $a,b,v$ do not build a bad triple in $\rho$, 
which is a contradiction. Summarizing, $\rho $ is an LDFS umbrella-free ordering.

Since $\widehat{\sigma }$ is an LDFS umbrella-free ordering of the vertices
of a path $P$, the ordering~$\widehat{\widehat{\sigma }}=$RMN$(\widehat{%
\sigma })$ gives a Hamiltonian path $P^{\prime }$ of the subgraph of $G$
induced by $V(P)$~\cite{Corneil-MPC}. Recall that $P^{\prime
}=(P_{0},z,x,P_{1},y,q,P_{2})$ for some subpaths $P_{0}$, $P_{1}$, and $%
P_{2} $ of $P^{\prime }$, where $V(P_{1})=A$. Thus, the graph $H$ induced by
the ordering $\rho $ of the vertices of $V(P)\cup \{v\}$ is again
Hamiltonian, since we can just insert in $P^{\prime }$ the new vertex $v$ of 
$\rho $ between $z$ and $x$. Therefore, since $\rho $ is an LDFS
umbrella-free ordering, the ordering $\widehat{\rho }=$RMN$(\rho )$ gives a
Hamiltonian path of $H$~\cite{Corneil-MPC}, i.e.~of the graph induced by $%
\rho $. We will compare now the orderings $\widehat{\widehat{\sigma }}$ and $%
\widehat{\rho }$.

First, we will prove that both orderings $\widehat{\widehat{\sigma }}$ and $%
\widehat{\rho }$ coincide until vertex $z$ is visited. Indeed, since $%
\widehat{\widehat{\sigma }}$ and $\widehat{\rho }$ differ only at the vertex 
$v$, the only difference of these orderings before $z$ is visited could be
that $v$ is visited before $z$ in $\widehat{\rho }$. Suppose that $v$ is
visited before $z$ in $\widehat{\rho }$. Note that the first vertex of the
ordering $\widehat{\rho }=$RMN$(\rho )$ is the rightmost vertex of $\rho $.
Therefore, $v$ is not the first vertex of $\widehat{\rho }$, since $v<_{\rho
}z$. Let $a$ be the previous vertex of $v$ in $\widehat{\rho }$. Then, $a$
is adjacent to $v$ in $\rho $, since $\widehat{\rho }$ is a path. If $a$ is
adjacent to $z$ in $\rho $, then $z$ is the next vertex of $a$ in $\widehat{%
\rho }$ instead of $v$, since $v<_{\rho }z$, which is a contradiction. Thus, 
$a$ is not adjacent to $z$ in both $\rho $ and $\widehat{\sigma }$. Note
that both orderings $\widehat{\widehat{\sigma }}$ and $\widehat{\rho }$
coincide at least until the visit of $a$, which is visited before $z$ in
both $\widehat{\widehat{\sigma }}$ and $\widehat{\rho }$, and thus $a\neq y$%
. If $a<_{\rho }y$ or $z<_{\rho }a$, it follows by the construction of $\rho 
$ that $a$ is adjacent also to $y$ in $\widehat{\sigma }$. Thus, since $v$
is the next vertex of $a$ in $\widehat{\rho }=$RMN$(\rho )$, it follows that 
$y$ is the next vertex of $a$ in $\widehat{\widehat{\sigma }}=$RMN$(\widehat{%
\sigma })$, i.e.~that $y$ is visited before $z$ in $\widehat{\widehat{\sigma 
}}$, which is a contradiction. Suppose that $v<_{\rho }a<_{\rho }z$. Then,
since $az\notin E$, it follows that $a\in A$, and thus $ay\in E$ in the
ordering~$\widehat{\sigma }$, as we proved above. Therefore, since $v$ is
the next vertex of $a$ in $\widehat{\rho }$, it follows that $y$ is the next
vertex of $a$ in $\widehat{\widehat{\sigma }}$, i.e.~that $y$ is visited
before $z$ in $\widehat{\widehat{\sigma }}$, which is a contradiction.
Therefore, $v$ is not visited before $z$ in $\widehat{\rho }$, and thus both
orderings $\widehat{\widehat{\sigma }}$ and $\widehat{\rho }$ coincide until
vertex $z$ is visited.

Now, $v$ is the rightmost unvisited neighbor of $z$ in $\rho $ at the time
that vertex $z$ is being visited by $\widehat{\rho }$, since by our initial
assumption $x$ is the next vertex of $z$ in $\widehat{\widehat{\sigma }}$.
Furthermore, similarly to $\widehat{\widehat{\sigma }}$, the ordering $%
\widehat{\rho }$ visits the vertices of $P_{1}$ after $v$, where $V(P_{1})=A$%
. In the sequel, after visiting all vertices of $P_{1}$, $\widehat{\rho }$
visits $y$ as the rightmost unvisited neighbor of the last vertex of $P_{1}$%
. Recall that $ys\notin E$ for every unvisited vertex $s$, such that $x<_{%
\widehat{\sigma }}s<_{\widehat{\sigma }}y$, and that the next vertex of $y$
in $\widehat{\widehat{\sigma }}$ is $q<_{\widehat{\sigma }}x$. Therefore, $x$
is the rightmost unvisited neighbor of $y$ in $\rho $ at the time that $y$
is being visited by $\widehat{\rho }$, and thus $\widehat{\rho }$ visits $x$
after $y$. Summarizing, the Hamiltonian path $P_{\rho }$ of the graph $H$
(i.e.~the graph induced by $\rho $) that is computed by $\widehat{\rho }$ 
is~$P_{\rho }=(P_{0},z,v,P_{1},y,x,Q)$ for some subpath $Q$ of $P_{\rho }$,
where $P^{\prime }=(P_{0},z,x,P_{1},y,q,P_{2})$. Note that $%
V(Q)=V(P_{2})\cup \{q\}$, since $V(P_{\rho })=V(P^{\prime })\cup
\{v\}=V(P)\cup \{v\}$. Furthermore, note that $Q$ is also a path of $G[V(P)]$%
, since $v\notin V(Q)$. Then, there exists the path $P^{\prime \prime
}=(P_{0},z,d,y,P_{1},x,Q)$ of $G$, where $V(P^{\prime \prime })=V(P)\cup
\{d\}$, which is a contradiction, since $P$ is a maximal path.

\emph{Case 2b.} $A=\emptyset $. Then $y$ is the next vertex of $x$ in $%
\widehat{\widehat{\sigma }}$, since $xy\in E$ and all vertices to the right
of $y$ in $\widehat{\sigma }$ have been already visited before $x$ in $%
\widehat{\widehat{\sigma }}$. That is, the path $P^{\prime }$ of the
vertices of~$V(P)$ constructed by $\widehat{\widehat{\sigma }}$ is~$%
P^{\prime }=(P_{0},z,x,y,P_{3})$, for some subpaths $P_{0}$ and $P_{3}$ of~$%
P^{\prime }$. Consider the ordering $\rho $, which obtained by adding a new
vertex $v$ to $\widehat{\sigma }$, as described in Case~2a. Then, similarly
to Case~2a, the graph $H$ induced by $\rho $ is Hamiltonian and the ordering 
$\widehat{\rho }=$RMN$(\rho )$ gives a Hamiltonian path $P_{\rho }$ of $H$,
where $P_{\rho }=(P_{0},z,v,y,Q)$. Note that $V(Q)=V(P_{3})\cup \{x\}$ and
that~$Q$ is also a path of $G[V(P)]$, since $v\notin V(Q)$. Thus, there
exists the path~${P^{\prime \prime }=(P_{0},z,d,y,Q)}$ of~$G$, where $%
V(P^{\prime \prime })=V(P)\cup \{d\}$, which is a contradiction, since $P$
is a maximal path. This completes the proof of the lemma.
\end{proof}

\medskip

The next lemma now follows by Lemma~\ref{typical-order}.

\begin{lemma}
\label{first-two}Let $x$ be the rightmost vertex in $\sigma ^{\prime }$ and $%
y$ be the rightmost neighbor of $x$ in $\sigma ^{\prime }$. Then,~$x$ is the
last vertex of $\widehat{\widehat{\sigma }}$ and $y$ is the previous vertex
of $x$ in $\widehat{\widehat{\sigma }}$.
\end{lemma}

\begin{proof}
First note that, if $\sigma ^{\prime }$ has at least two vertices, $x$ is
not the first vertex of $\widehat{\widehat{\sigma }}$, since~$\widehat{%
\widehat{\sigma }}=$RMN$(\widehat{\sigma })$ and $x$ is the leftmost vertex
of $\widehat{\sigma }$. Suppose that $x$ is not the last vertex of $\widehat{%
\widehat{\sigma }}$, i.e.~$x$ is an intermediate vertex. Let $a$ and $b$ be
the previous and the next vertices of~$x$ in~$\widehat{\widehat{\sigma }}$,
respectively. Then, $a<_{\widehat{\widehat{\sigma }}}b$. If $ab\in E$, then $%
b$ is the next vertex of $a$ in $\widehat{\widehat{\sigma }}$ instead of~$x$, 
since~$x<_{\widehat{\sigma }}a$, which is a contradiction. Therefore~$%
ab\notin E$, and thus $b<_{\widehat{\sigma }}a$ by Lemma~\ref{auxil-2},
since~$a<_{\widehat{\widehat{\sigma }}}b$. Furthermore $a<_{\sigma ^{\prime
}}b$ by Lemma~\ref{auxil-2}, since $b<_{\widehat{\sigma }}a$ and $ab\notin E$%
, and thus~$a<_{\sigma ^{\prime }}b<_{\sigma ^{\prime }}x$, since $x$ is the
rightmost vertex in $\sigma ^{\prime }$. That is, $x<_{\widehat{\sigma }}b<_{%
\widehat{\sigma }}a$ and $a<_{\sigma ^{\prime }}b<_{\sigma ^{\prime }}x$,
where $xb,xa\in E$ and $ba\notin E$, while $x$ is the next vertex of $a$ in $%
\widehat{\widehat{\sigma }}$, which is a contradiction by Lemma~\ref%
{typical-order}. Therefore,~$x$ is the last vertex of $\widehat{\widehat{%
\sigma }}$.

Note now that $y$ is the second leftmost vertex in $\widehat{\sigma }$,
since $x$ is the rightmost vertex of $\sigma ^{\prime }$ and~$\widehat{%
\sigma }=$LDFS$^{+}(\sigma ^{\prime })$. Suppose that $y$ is not the
previous vertex of $x$ in $\widehat{\widehat{\sigma }}$ and let $a\neq y$ be
the previous vertex of $x$ in $\widehat{\widehat{\sigma }}$. Then, $x<_{%
\widehat{\sigma }}y<_{\widehat{\sigma }}a$ and $xy,xa\in E$. Furthermore, $y$
has been visited before $a$ in $\widehat{\widehat{\sigma }}$, i.e.~$y<_{%
\widehat{\widehat{\sigma }}}a$, since $x$ is the last vertex of $\widehat{%
\widehat{\sigma }}$. Suppose that $ya\notin E$. Then, since~$y<_{\widehat{%
\sigma }}a$, it follows by Lemma~\ref{auxil-2} that $a<_{\widehat{\widehat{%
\sigma }}}y$, which is a contradiction, since $y<_{\widehat{\widehat{\sigma }%
}}a$. Therefore $ya\in E$. Thus, since $y<_{\widehat{\sigma }}a$ and~$y<_{%
\widehat{\widehat{\sigma }}}a$, Lemma~\ref{auxil-3} implies that $y$ is not
the first vertex of $\widehat{\widehat{\sigma }}$ and that $y<_{\widehat{%
\sigma }}a<_{\widehat{\sigma }}z$, $yz\in E$, and $az\notin E$ for the
previous vertex $z$ of $y$ in~$\widehat{\widehat{\sigma }}$. Furthermore $%
z<_{\sigma ^{\prime }}a$ by Lemma~\ref{auxil-2}, since $a<_{\widehat{\sigma }%
}z$ and $az\notin E$. On the other hand $a<_{\sigma ^{\prime }}y$, 
since~$xa\in E$ and $y$ is the rightmost neighbor of $x$ in $\sigma ^{\prime }$.
That is, $y<_{\widehat{\sigma }}a<_{\widehat{\sigma }}z$ and $z<_{\sigma
^{\prime }}a<_{\sigma ^{\prime }}y$, where~$ya,yz\in E$ and~$az\notin E$,
while~$y$ is the next vertex of $z$ in $\widehat{\widehat{\sigma }}$, which
is a contradiction by~Lemma~\ref{typical-order}. Therefore,~$y$ is the
previous vertex of $x$ in $\widehat{\widehat{\sigma }}$.
\end{proof}

\medskip

The next corollary follows easily by Definition~\ref{normal-def}(a) and
Lemma~\ref{first-two}.

\begin{corollary}
\label{typical-cor}Let $G=(V,E)$ be a cocomparability graph, $\sigma $ be 
an LDFS umbrella-free ordering of $G$, and $P$ be a maximal path of $G$. 
Then there exists a typical path $P^{\prime }$ of $G$, such that~${V(P^{\prime})=V(P)}$.
\end{corollary}

\begin{proof}
Consider the restriction ${\sigma^{\prime }=\sigma |_{V(P)}}$ of $\sigma$ on the vertices of $P$; 
note that $\sigma^{\prime}$ is an induced subordering of $\sigma$. 
Furthermore, consider the orderings~$\widehat{\sigma }=$LDFS$^{+}(\sigma ^{\prime })$ 
and~$\widehat{\widehat{\sigma }}=$RMN$(\widehat{\sigma})$ (cf.~Notation~\ref{maximal-path-orderings-notation}). 
Note that, since $\sigma^{\prime}$ is an ordering of the vertices of $V(P)$, 
the ordering $\widehat{\widehat{\sigma }}$ defines a minimum path cover of $G[V(P)]$~\cite{Corneil-MPC}. 
Therefore, since $G[V(P)]$ has $P$ as a Hamiltonian path, it follows that 
the ordering $\widehat{\widehat{\sigma }}$ defines a single path $Q$ on the vertices of $V(P)$ 
(note that this path~$Q$ may be~$P$ itself or a different path on the same vertices). 
Let now $x$ be the rightmost vertex in $\sigma^{\prime}$ and 
$y$ be the rightmost neighbor of $x$ in $\sigma^{\prime}$. 
Then, since $\widehat{\widehat{\sigma }}$ defines the path $Q$, 
Lemma~\ref{first-two} implies that~$x$ is the last vertex of $Q$ 
and $y$ is the previous vertex of $x$ in $Q$. 
Therefore, the \emph{reverse} path $P^{\prime}$ of~$Q$ is a typical path of $G$ with $V(P^{\prime})=V(P)$.
\end{proof}

\medskip

We are now ready to present the main theorem of this section.

\begin{theorem}
\label{normal-thm}Let $G=(V,E)$ be a cocomparability graph, $\sigma $ be an
LDFS umbrella-free ordering of $G$, and $P$ be a maximal path of $G$. Then
there exists a normal path $P^{\prime }$ of $G$, such that~${V(P^{\prime})=V(P)}$.
\end{theorem}

\begin{proof}
Let the maximal path $P$ be denoted $(v_{1},v_{2},\ldots ,v_{k})$.%
If $k\leq 2$, the lemma clearly holds. Suppose in the sequel that $k\geq 3$
and that there exists no normal path $P^{\prime }$ of $G$, such 
that~${V(P^{\prime})=V(P)}$. We may assume without loss of generality that $G$ has
the smallest number of vertices among all cocomparability graphs that have
such a maximal path $P$. Furthermore, we may assume by Corollary~\ref%
{typical-cor} that $P$ is typical, i.e.~that $v_{1}$ is the rightmost vertex
of $V(P)$ in~$\sigma $ and that $v_{2}$ is the rightmost vertex of $%
N(v_{1})\cap \{v_{2},v_{3},\ldots ,v_{k}\}$ in~$\sigma $.

Let $i\in \{2,3,\ldots ,k-1\}$ be the greatest index, such that $v_{j}$ is
the rightmost vertex of $N(v_{j-1})\cap \{v_{j},v_{j+1},\ldots ,v_{k}\}$ in~$%
\sigma $ for every $j=2,\ldots ,i$. Such an index $i$ exists by the
assumption that there exists no normal path $P^{\prime }$ of $G$, for which $%
V(P^{\prime })=V(P)$. Let $P_{1}=(v_{1},v_{2},\ldots ,v_{i})$ and $%
P_{2}=(v_{i+1},v_{i+2},\ldots ,v_{k})$ be the subpaths of $P$ until the
vertex $v_{i}$ and after the vertex $v_{i}$, respectively. Then, in
particular, $P_{1}$ is normal by the assumption on $i$, i.e.~$P_{1}$ has the
first $i$ vertices of an RMN when applied on the restriction $\sigma
|_{V(P)} $ of the ordering $\sigma $ on the vertices of~$P$. We will
construct a path $P^{\ast }=(v_{1}^{\ast },v_{2}^{\ast },\ldots ,v_{k}^{\ast
})$, such that $V(P^{\ast })=V(P)$, $v_{1}^{\ast }$ is the rightmost vertex
of $V(P^{\ast })$ in $\sigma $, and $v_{\ell }^{\ast }$ is the rightmost
vertex of $N(v_{\ell -1}^{\ast })\cap \{v_{\ell }^{\ast },v_{\ell +1}^{\ast
},\ldots ,v_{k}^{\ast }\}$ in~$\sigma $ for every $\ell =2,\ldots ,i+1$,
thus arriving to a contradiction by the assumption on the index~$i$.

Consider a vertex $v_{\ell }\in \{v_{i+1},v_{i+2},\ldots ,v_{k}\}$, such
that $v_{i}<_{\sigma }v_{\ell }$. Then, $v_{i}<_{\sigma }v_{\ell }<_{\sigma
}v_{1}$, since~$P$ is typical. We will prove that $v_{i}v_{\ell }\in E$.
Suppose otherwise that $v_{i}v_{\ell }\notin E$. Then, since $v_{\ell
}\notin V(P_{1})$, it follows by Lemma~\ref{auxil} that there exist two
consecutive vertices $v_{j-1}$ and $v_{j}$ in $P_{1}$, where $2\leq j\leq i$%
, such that $v_{j-1}v_{\ell }\in E$ and $v_{j}<_{\pi }v_{\ell }$. Thus, $%
v_{j}$ is not the rightmost vertex of $N(v_{j-1})\cap \{v_{j},v_{j+1},\ldots
,v_{k}\}$ in~$\sigma $, which is a contradiction. Therefore, $v_{i}v_{\ell
}\in E$ for every $v_{\ell }\in \{v_{i+1},v_{i+2},\ldots ,v_{k}\}$, such
that $v_{i}<_{\sigma }v_{\ell }$.

In the sequel let $v_{j}$ be the rightmost vertex of $N(v_{i})\cap
\{v_{i+1},v_{i+2},\ldots ,v_{k}\}$ in~$\sigma $, where $j>i+1$ by the
assumption on the index $i$. Now we distinguish the cases where $%
v_{i}<_{\sigma }v_{j}$ and $v_{j}<_{\sigma }v_{i}$.

\medskip

\emph{Case 1.} $v_{i}<_{\sigma }v_{j}$. Suppose that there exists a vertex $%
v_{\ell }\in \{v_{i+1},v_{i+2},\ldots ,v_{k}\}$, such that $v_{j}<_{\sigma
}v_{\ell }$. Then, as we proved above, $v_{i}v_{\ell }\in E$, which is a
contradiction, since $v_{j}$ is the rightmost vertex of $N(v_{i})\cap
\{v_{i+1},v_{i+2},\ldots ,v_{k}\}$ in $\sigma $ and $v_{j}<_{\sigma }v_{\ell
}$. Thus, $v_{j}$ is the rightmost vertex of $\{v_{i+1},v_{i+2},\ldots
,v_{k}\}$ in~$\sigma $. Let $\sigma ^{\prime }$ be the induced subordering
of~$\sigma $ on the vertices of $V(P_{2})=\{v_{i+1},v_{i+2},\ldots ,v_{k}\}$%
, and $\sigma ^{\prime \prime }$ be an LDFS closure of~$\sigma ^{\prime }$
(within $\sigma $). Then, by definition $V(P_{1})\cap V(\sigma ^{\prime
})=\emptyset $. Furthermore, $v_{j}$ is the rightmost vertex in~$\sigma
^{\prime }$, and thus $v_{j}$ remains the rightmost vertex in $\sigma
^{\prime \prime }$ by Corollary~\ref{rightmost-in-closure}.

First we will prove that ${V(P_{1})\cap V(\sigma ^{\prime \prime })=\emptyset}$. 
Suppose otherwise that~${V(P_{1})\cap V(\sigma ^{\prime \prime })\neq \emptyset}$, 
and let $v$ be the rightmost vertex of $V(P_{1})\cap V(\sigma
^{\prime \prime })$ in $\sigma $. Then $v\in V(\sigma ^{\prime \prime
}\setminus \sigma ^{\prime })$, since~${V(P_{1})\cap V(\sigma ^{\prime
})=\emptyset}$. Thus, there exists by Lemma~\ref{infinite-sequence-closure}
at least one vertex $v^{\prime }$ in $\sigma ^{\prime }$, such that $%
v<_{\sigma }v^{\prime }$ and~${vv^{\prime }\notin E}$. 
Then,~${v^{\prime }\in V(P_{2})\subset V(P)}$ by definition of the ordering $\sigma ^{\prime }$,
i.e.~$v^{\prime }$ is a vertex of $\sigma |_{V(P)}$. Furthermore, since $%
v<_{\sigma }v^{\prime }$ and $vv^{\prime }\notin E$, Lemma~\ref{auxil-2}
implies that $v^{\prime }$ is visited before $v$ in $P_{1}$ (that is, by
applying RMN on $\sigma |_{V(P)}$), i.e.~$v^{\prime }\in V(P_{1})$, which is
a contradiction, since $v^{\prime }\in V(P_{2})$. 
Thus,~$V(P_{1})\cap V(\sigma ^{\prime \prime })=\emptyset $.

Now we will prove that the subpath $P_{2}=(v_{i+1},v_{i+2},\ldots ,v_{k})$
of $P$ is maximal in $G|_{\sigma ^{\prime \prime }}$. Indeed, suppose
otherwise that $P_{2}$ is not maximal in $G|_{\sigma ^{\prime \prime }}$,
i.e.~there exists a path $P_{2}^{\prime }$ of $G|_{\sigma ^{\prime \prime }}$%
, such that $V(P_{2})\subset V(P_{2}^{\prime })$. Thus, since $G|_{\sigma
^{\prime \prime }}$ has strictly fewer vertices than $G$, there exists (by
the assumption on $G$) a normal path $P_{2}^{\prime \prime }$ of $G|_{\sigma
^{\prime \prime }}$, such that $V(P_{2}^{\prime \prime })=V(P_{2}^{\prime })$%
. Therefore, in particular, $P_{2}^{\prime \prime }$ has strictly more
vertices than $P_{2}$ and $v_{j}$ is the first vertex of $P_{2}^{\prime
\prime }$. Thus, since $v_{i}v_{j}\in E$, the path $(v_{1},v_{2},\ldots
,v_{i},P_{2}^{\prime \prime })$ of $G$ has strictly more vertices than $P$,
which is a contradiction to the assumption that $P$ is maximal. Therefore,
the subpath $P_{2}$ of $P$ is maximal in $G|_{\sigma ^{\prime \prime }}$,
and thus there exists a normal path $Q$ of $G|_{\sigma ^{\prime \prime }}$,
such that $V(Q)=V(P_{2})$. Then, in particular, $v_{j}$ is the first vertex
of $Q$, and thus%
\begin{equation}
P^{\ast }=(v_{1}^{\ast },v_{2}^{\ast },\ldots ,v_{k}^{\ast
})=(v_{1},v_{2},\ldots ,v_{i},Q)
\end{equation}%
is as requested.

\medskip

\emph{Case 2.} $v_{j}<_{\sigma }v_{i}$. Consider an arbitrary vertex $%
v_{\ell }\in \{v_{i+1},v_{i+2},\ldots ,v_{k}\}$ and suppose that $%
v_{i}<_{\sigma }v_{\ell }$. Then, $v_{\ell }\neq v_{j}$, since $%
v_{j}<_{\sigma }v_{i}$. Furthermore, as we proved above, $v_{i}v_{\ell }\in
E $, which is a contradiction, since $v_{j}$ is the rightmost vertex of $%
N(v_{i})\cap \{v_{i+1},v_{i+2},\ldots ,v_{k}\}$ in~$\sigma $ and $%
v_{j}<_{\sigma }v_{i}<_{\sigma }v_{\ell }$. Therefore, $v_{\ell }<_{\sigma
}v_{i}$ for every $v_{\ell }\in \{v_{i+1},v_{i+2},\ldots ,v_{k}\}$, i.e.~$%
v_{i}$ is the rightmost vertex of $V(P_{2})\cup
\{v_{i}\}=\{v_{i},v_{i+1},\ldots ,v_{k}\}$ in $\sigma $. Consider the
induced subordering $\sigma ^{\prime }$ of~$\sigma $ on the vertices $%
V(P_{2})\cup \{v_{i}\}$ and an LDFS closure $\sigma ^{\prime \prime }$ of~$%
\sigma ^{\prime }$ (within $\sigma $). Then, similarly to Case 1, the
subpath $(v_{i},P_{2})=(v_{i},v_{i+1},\ldots ,v_{k})$ of $P$ is a maximal
path of $G|_{\sigma ^{\prime \prime }}$, and thus there exists a normal path 
$Q$ of $G|_{\sigma ^{\prime \prime }}$, such that $V(Q)=\{v_{i},v_{i+1},%
\ldots ,v_{k}\}$. Then, in particular, $v_{i}$ is the first vertex of $Q$,
and thus%
\begin{equation}
P^{\ast }=(v_{1}^{\ast },v_{2}^{\ast },\ldots ,v_{k}^{\ast
})=(v_{1},v_{2},\ldots ,v_{i-1},Q)
\end{equation}%
is as requested. This completes the proof of the lemma.
\end{proof}

\section{The longest path problem on cocomparability graphs\label%
{sec:longest}}

In this section we present the first polynomial algorithm that computes a
longest path of a cocomparability graph $G$. This dynamic programming
algorithm is based on Theorem~\ref{normal-thm}; in particular, this
algorithm computes a \emph{longest normal} path of $G$. For the rest of this
section we consider an LDFS umbrella-free ordering $\sigma $ of a given
cocomparability graph $G=(V,E)$, which can be obtained by executing an LDFS$%
^{+}$ on an arbitrary umbrella-free ordering $\pi $ of $G$~\cite{Corneil-MPC}. 
We consider that the vertices of $V$, where $|V|=n$, are numbered in $%
\sigma $ increasingly from left to right, i.e.~$\sigma =(u_{1},u_{2},\ldots
,u_{n})$. Furthermore, for simplicity of the presentation, we add to $\sigma 
$ a dummy isolated vertex $u_{n+1}$ to the right of all other vertices of $V$, 
i.e.~we consider without loss of generality that $\sigma
=(u_{1},u_{2},\ldots ,u_{n},u_{n+1})$. It is easy to see that $\sigma $
remains an LDFS umbrella-free ordering after the addition of the dummy
vertex $u_{n+1}$.

\begin{definition}
\label{G(i,j)}Let ${G=(V,E)}$ be a cocomparability graph with ${|V|=n}$ and
let ${\sigma =(u_{1},u_{2},\ldots ,u_{n},u_{n+1})}$ be an LDFS umbrella-free
ordering of $V\cup \{u_{n+1}\}$, where $u_{n+1}$ is a dummy isolated vertex.
For every pair of indices $i,j\in \{1,2,\ldots ,n\}$, \vspace{-0.2cm}
\begin{itemize}
\item if $i>j$, then $G(i,j)=\emptyset $, \vspace{-0.2cm}
\item if ${i\leq j}$, then $G(i,j)$ is the subgraph $G[S]$ of $G$ induced by
the vertex set ${S=\{u_{i},u_{i+1},\ldots ,u_{j}\}\setminus N(u_{j+1})}$.
\end{itemize}
\end{definition}

It is easy to see by Definition~\ref{G(i,j)} that for every pair of indices 
$i,j\in \{1,2,\ldots ,n\}$, the vertices~$u_{i}$ and~$u_{j}$ may or may not
belong to $G(i,j)$, since they may or may not be adjacent to~$u_{j+1}$ in~$G$. 
Furthermore note that $G(1,n)=G$ and that $G(i,n)=G[{\{u_{i},u_{i+1},\ldots ,u_{n}\}}]$ 
for every~${i\in \{1,2,\ldots ,n\}}$, since $u_{n+1}$ is an isolated vertex.

As an example of Definition~\ref{G(i,j)}, the subgraph $G(3,8)$ of the cocomparability graph $G$ 
of~Figure~\ref{sigma-fig} is illustrated in Figure~\ref{Gij-subgraph-fig}. 
In this figure the dummy isolated vertex $u_{10}$ is also depicted, 
while the vertices $V(G(3,8))=\{u_{3},u_{4},u_{5},u_{6},u_{8}\}$ of $G(3,8)$, 
as well as the edges of~$G(3,8)$, are drawn darker than the others for better visibility. 
Furthermore, note that the path $P=(u_{8},u_{5},u_{6},u_{3},u_{4})$ is a normal path of~$G(3,8)$.

\begin{figure}[h!]
\centering
\includegraphics[scale=0.9]{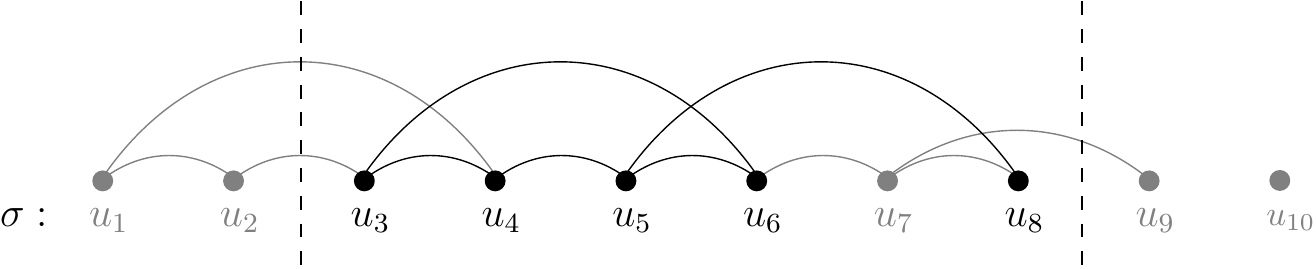} 
\caption{The subgraph $G(3,8)$ of the cocomparability graph $G$ of Figure~\ref{sigma-fig}.}
\label{Gij-subgraph-fig}
\end{figure}

\begin{observation}
\label{G(i,j)-u_i}For every pair of indices $i,j\in \{1,2,\ldots ,n\}$, $%
G(i+1,j)=G(i,j)\setminus \{u_{i}\}$.
\end{observation}

\begin{observation}
\label{subpath-normal-obs}Let ${P=(P_{1},u_{i})}$ be a normal path of $G(i,j)$, 
for some pair of indices ${i,j\in \{1,2,\ldots ,n\}}$. Then $P_{1}$ is a
normal path of both $G(i+1,j)$ and $G(i,j)$.
\end{observation}

\begin{observation}
\label{superpath-normal-obs}Let $P_{1}=(P_{0},u_{x})$ be a normal path of $%
G(i+1,j)$, for some pair of indices $i,j\in \{1,2,\ldots ,n\}$, and let $%
u_{i}\in V(G(i,j))$ and $u_{x}\in N(u_{i})$. Then $P=(P_{1},u_{i})$ is a
normal path of $G(i,j)$.
\end{observation}

\begin{lemma}
\label{u_i-G(i+1,x-1)}Let ${G=(V,E)}$ be a cocomparability graph and $\sigma
=(u_{1},u_{2},\ldots ,u_{n},u_{n+1})$ be an LDFS umbrella-free ordering of $%
V\cup \{u_{n+1}\}$, where $u_{n+1}$ is a dummy isolated vertex. 
Suppose that~${u_{i}<_{\sigma}u_{x}}$ and $u_{x}\in N(u_{i})$. 
Then $u_{k}\in N(u_{i})$ for every $u_{k}\in V(G(i+1,x-1))$.
\end{lemma}

\begin{proof}
Let $u_{k}\in V(G(i+1,x-1))$. Then $u_{i}<_{\sigma }u_{k}<_{\sigma }u_{x}$
and $u_{k}\notin N(u_{x})$ by Definition~\ref{G(i,j)}. Therefore, since $%
\sigma $ is an umbrella-free ordering and $u_{x}\in N(u_{i})$ by assumption,
it follows that~$u_{k}\in N(u_{i})$.
\end{proof}

\begin{lemma}
\label{G(i+1,x-1)}Let ${G=(V,E)}$ be a cocomparability graph and $\sigma
=(u_{1},u_{2},\ldots ,u_{n},u_{n+1})$ be an LDFS umbrella-free ordering of $%
V\cup \{u_{n+1}\}$, where $u_{n+1}$ is a dummy isolated vertex. Then ${%
V(G(i+1,x-1))\subseteq V(G(i,j))}$ for every $u_{x}\in V(G(i+1,j))$.
\end{lemma}

\begin{proof}
Consider a vertex $u_{y}\in V(G(i+1,x-1))$. Then, since also $u_{x}\in
V(G(i+1,j))$, it follows by Definition~\ref{G(i,j)} that $u_{y}\notin
N(u_{x})$ and $u_{x}\notin N(u_{j+1})$. Suppose that $u_{y}\in N(u_{j+1})$.
Then, since $u_{y}<_{\sigma }u_{x}<_{\sigma }u_{j+1}$, the vertices $%
u_{y},u_{x},u_{j+1}$ build an umbrella in $\sigma $, which is a
contradiction. Therefore $u_{y}\notin N(u_{j+1})$, and thus $u_{y}\in
V(G(i,j))$ by Definition~\ref{G(i,j)}.
\end{proof}

\medskip

In the following we state two lemmas that are crucial for the proof of the
main Theorem~\ref{correctness} of this section.

\begin{lemma}
\label{P=(P1,ui,P2)-normal}Let ${G=(V,E)}$ be a cocomparability graph and ${%
\sigma =(u_{1},u_{2},\ldots ,u_{n},u_{n+1})}$ be an LDFS umbrella-free
ordering of ${V\cup \{u_{n+1}\}}$, where $u_{n+1}$ is a dummy isolated
vertex. Let~${u_{i}\in V(G(i,j))}$, ${u_{x}\in V(G(i+1,j))}$, ${u_{y}\in
V(G(i+1,x-1))}$, and ${u_{x}\in N(u_{i})}$. Furthermore, let $P_{1}$ be a
normal path of~$G(i+1,j)$ with $u_{x}$ as its last vertex and $P_{2}$ be a
normal path of~$G(i+1,x-1)$ with $u_{y}$ as its last vertex. Then ${%
P=(P_{1},u_{i},P_{2})}$ is a normal path of~$G(i,j)$ with~$u_{y}$ as its
last vertex.
\end{lemma}

\begin{proof}
We will first prove that $V(P_{1})\subseteq V(G(i+1,j))\setminus
V(G(i+1,x-1))$. Suppose otherwise that $V(P_{1})\cap V(G(i+1,x-1))\neq
\emptyset $, and let $u_{k}$ be the first vertex of $P_{1}$, such that $%
u_{k}\in V(G(i+1,x-1))$. Then $u_{k}$ is not the rightmost vertex of $P_{1}$
in $\sigma $, since $u_{k}<_{\sigma }u_{x}$. Therefore, since $P_{1}$ is a
normal path by assumption, $u_{k}$ is not the first vertex of $P_{1}$, and
thus there exists a previous vertex $u_{\ell }$ of $u_{k}$ in~$P_{1}$, i.e.~$%
u_{\ell }\in N(u_{k})$. Suppose first that $u_{\ell }\in N(u_{x})$. Then,
since $u_{k}<_{\sigma }u_{x}$ and $u_{x}$ is unvisited by $P_{1}$ when $%
u_{\ell }$ is visited, it follows that $u_{k}$ is not the rightmost
unvisited vertex of $N(u_{\ell })\cap V(P_{1})$ in~$\sigma $, when~$P_{1}$
visits $u_{\ell }$. This is a contradiction by Definition~\ref{normal-def},
since $u_{k}$ is the next vertex of $u_{\ell }$ in $P_{1}$ and~$P_{1}$ is a
normal path by assumption. Suppose now that $u_{\ell }\notin N(u_{x})$. Let $%
u_{\ell }<_{\sigma }u_{x}$. Then $u_{\ell }\in V(G(i+1,x-1))$ by Definition~%
\ref{G(i,j)}. This is a contradiction to the assumption that $u_{k}$ is the
first vertex of $P_{1}$, such that $u_{k}\in V(G(i+1,x-1))$. Let $%
u_{x}<_{\sigma }u_{\ell }$, i.e.~$u_{k}<_{\sigma }u_{x}<_{\sigma }u_{\ell }$%
. Note that $u_{k}\notin N(u_{x})$ by Definition~\ref{G(i,j)}, since $%
u_{k}\in V(G(i+1,x-1))$. Thus the vertices $u_{k},u_{x},u_{\ell }$ build an
umbrella in $\sigma $, since $u_{\ell }\in N(u_{k})$, $u_{k}\notin N(u_{x})$%
, and $u_{\ell }\notin N(u_{x})$, which is a contradiction. Therefore $%
V(P_{1})\cap V(G(i+1,x-1))=\emptyset $, i.e.~$V(P_{1})\subseteq
V(G(i+1,j))\setminus V(G(i+1,x-1))$.

Since $V(P_{1})\subseteq V(G(i+1,j))\setminus V(G(i+1,x-1))$ by the previous
paragraph and ${V(P_{2})\subseteq V(G(i+1,x-1))}$ by assumption, it follows
that ${V(P_{1})\cap V(P_{2})=\emptyset}$. Recall now that ${u_{k}\in N(u_{i})%
}$ for every ${u_{k}\in V(P_{2})\subseteq V(G(i+1,x-1))}$ by Lemma~\ref%
{u_i-G(i+1,x-1)}. Furthermore, recall that~${V(P_{1})\subseteq
V(G(i+1,j))\subseteq V(G(i,j))}$ by Observation~\ref{G(i,j)-u_i} and that~${%
V(P_{2})\subseteq V(G(i+1,x-1))\subseteq V(G(i,j))}$ by Lemma~\ref%
{G(i+1,x-1)}. Therefore, since $u_{i}\in V(G(i,j))$ and $u_{x}\in N(u_{i})$
by assumption, it follows that $P=(P_{1},u_{i},P_{2})$ is a path of $G(i,j)$%
. Moreover~$u_{y}$ is the last vertex of $P$, since $u_{y}$ is the last
vertex of $P_{2}$ by assumption.

In the following we prove that $P$ is normal. To this end, first let $\sigma
_{1}=\sigma |_{P_{1}}$ be the restriction of the ordering $\sigma $ on the
vertices of the path $P_{1}$ and let $\sigma _{1}^{\prime }=$RMN$(\sigma
_{1})$. Then the ordering of the vertices of $V(P_{1})$ in $P_{1}$ coincides
with the ordering $\sigma _{1}^{\prime }$ by Observation~\ref{normal-RMN}.
Note that $\sigma _{1}$ is an umbrella-free ordering, as a restriction of
the umbrella-free ordering $\sigma $.

Note now that the first vertex $u_{\ell }$ of $P$ is also the first vertex
of $P_{1}$, since ${P=(P_{1},u_{i},P_{2})}$. Moreover, $u_{\ell }$ is the
rightmost vertex of $P_{1}$ in $\sigma $, since $P_{1}$ is normal by
assumption. Furthermore, note that $u_{k}<_{\sigma }u_{x}\leq _{\sigma
}u_{\ell }$ for every $u_{k}\in V(P_{2})\cup \{u_{i}\}$. Therefore, $u_{\ell
}$ is also the rightmost vertex of $P$ in $\sigma $. Let $u_{r}$ and $%
u_{r^{\prime }}$ be two consecutive vertices of $P_{1}$, i.e.~$u_{r^{\prime
}}$ is the rightmost unvisited vertex of $N(u_{r})\cap V(P_{1})$ in~$\sigma $%
, when $P_{1}$ visits $u_{r}$. We will prove that $u_{r^{\prime }}$ is also
the rightmost unvisited vertex of $N(u_{r})\cap V(P)$ in~$\sigma $, when $P$
visits $u_{r}$. Suppose otherwise that $u_{k}\neq u_{r^{\prime }}$ is the
rightmost unvisited vertex of $N(u_{r})\cap V(P)$ in $\sigma $, when $P$
visits $u_{r}$. Then in particular $u_{r^{\prime }}<_{\sigma }u_{k}$ and $%
u_{k}\in N(u_{r})$. If $u_{k}\in V(P_{1})$, then $u_{k}$ would be also the
rightmost unvisited vertex of $N(u_{r})\cap V(P_{1})$ in $\sigma $, when $%
P_{1}$ visits $u_{r}$, which is a contradiction.

Therefore $u_{k}\in V(P_{2})\cup \{u_{i}\}\subseteq \{u_{i},u_{i+1},\ldots
,u_{x-1}\}$, and thus in particular $u_{k}<_{\sigma }u_{x}$. 
Suppose that $u_{r}\in N(u_{x})$. Then, since $u_{k}<_{\sigma }u_{x}$ and $u_{x}$ is
unvisited when $P$ visits $u_{r}$, it follows that $u_{k}$ is not the
rightmost unvisited vertex of $N(u_{r})\cap V(P)$ in~$\sigma $ when $P$
visits~$u_{r}$, which is a contradiction to the assumption on $u_{k}$.
Thus $u_{r}\notin N(u_{x})$. Recall that $u_{x}$ is the last
vertex of $P_{1}$ by assumption. Therefore, $u_{r}$ appears before $u_{x}$
in $P_{1}$, and thus $u_{r}<_{\sigma _{1}^{\prime }}u_{x}$ as we proved
above, where $\sigma _{1}=\sigma |_{P_{1}}$ and $\sigma _{1}^{\prime }=$RMN$%
(\sigma _{1})$. Therefore, since $u_{r}\notin N(u_{x})$ and $\sigma _{1}$ is
an umbrella-free ordering, it follows by Lemma~\ref{auxil-2} that $%
u_{x}<_{\sigma _{1}}u_{r}$, i.e.~$u_{x}<_{\sigma }u_{r}$. That is, $%
u_{k}<_{\sigma }u_{x}<_{\sigma }u_{r}$. Recall that $u_{k}\in V(P_{2})\cup
\{u_{i}\}$. First let $u_{k}\in V(P_{2})\subseteq V(G(i+1,x-1))$. Then $%
u_{k}\notin N(u_{x})$ by Definition~\ref{G(i,j)}. Therefore, since also $%
u_{r}\notin N(u_{x})$ and $u_{k}\in N(u_{r})$, the vertices $%
u_{k},u_{x},u_{r}$ build an umbrella in $\sigma $, which is a contradiction.
Now let $u_{k}=u_{i}$. Then~$u_{k}=u_{i}<_{\sigma }u_{r^{\prime }}$, since~$%
u_{r^{\prime }}\in V(P_{1})\subseteq V(G(i+1,j))$. Thus $u_{k}$ is not the
rightmost unvisited vertex of $N(u_{r})\cap V(P)$ in~$\sigma $, when $P$
visits $u_{r}$, which is a contradiction to the assumption on $u_{k}$.
Therefore, for any two consecutive vertices $u_{r},u_{r^{\prime }}$ of $%
P_{1}$, $u_{r^{\prime }}$ is the rightmost unvisited vertex of~$N(u_{r})\cap
V(P)$ in~$\sigma $, when $P$ visits~$u_{r}$.

Recall that $V(P_{2})\subseteq V(G(i+1,x-1))$ by assumption, and thus $%
u_{k}\notin N(u_{x})$ for every vertex~$u_{k}\in V(P_{2})$. Therefore, $%
u_{i} $ is the rightmost unvisited vertex of $N(u_{x})\cap V(P)$ in $\sigma $%
, when~$P$ visits $u_{x}$ (i.e.~the last vertex of $P_{1}$). Note that
exactly the vertices of $V(P_{2})$ are the unvisited vertices of $V(P)$,
when $P$ visits $u_{i}$. Moreover, recall that $P_{2}$ is a normal path and
that~$u_{k}\in N(u_{i})$ for every $u_{k}\in V(P_{2})\subseteq V(G(i+1,x-1))$
by Lemma~\ref{u_i-G(i+1,x-1)}. Therefore, the first vertex of $P_{2}$ is
also the rightmost unvisited vertex of $N(u_{i})\cap V(P)$ in $\sigma $,
when $P$ visits $u_{i}$. Consider now any pair of consecutive vertices $%
u_{r},u_{r^{\prime }}$ of $P_{2}$. Then, $u_{r^{\prime }}$ is the rightmost
unvisited vertex of $N(u_{r})\cap V(P_{2})$ in $\sigma $ (resp.~of $%
N(u_{r})\cap V(P)$ in $\sigma $), when $P_{2}$ (resp.~$P$) visits~$u_{r}$.
Therefore,~$P$ is a normal path. This completes the proof of the lemma.
\end{proof}

\begin{notation}
\label{longest-normal-with-end-vertex-notation}
Let ${G=(V,E)}$ be a cocomparability graph and ${\sigma=(u_{1},u_{2},\ldots,u_{n},u_{n+1})}$ 
be an LDFS umbrella-free ordering of~${V\cup \{u_{n+1}\}}$, where~$u_{n+1}$ is a dummy isolated vertex. 
Let~${i,j \in \{1,2,\ldots,n\}}$ be a pair of indices, let~$u_{k} \in V(G(i,j))$, and let~$P$ be a normal path of~$G(i,j)$. 
For simplicity of presentation, we will say in the following that 
``$P$ is a longest normal path of~$G(i,j)$ with~$u_{k}$ as its last vertex'' 
if $P$ has the greatest number of vertices among those normal paths of $G(i,j)$ that have $u_{k}$ as their last vertex.
\end{notation}

\begin{lemma}
\label{P1,P2-normal}Let ${G=(V,E)}$ be a cocomparability graph and ${\sigma
=(u_{1},u_{2},\ldots ,u_{n},u_{n+1})}$ be an LDFS umbrella-free ordering of $%
{V\cup \{u_{n+1}\}}$, where $u_{n+1}$ is a dummy isolated vertex. Let~$P$ be
a longest normal path of~$G(i,j)$ with $u_{y}\neq u_{i}$ as its last vertex
and let $P=(P_{1},u_{i},P_{2})$. Let~$u_{x}$ be the last vertex of $P_{1}$.
Then $P_{1}$ is a longest normal path of $G(i+1,j)$ with $u_{x}$ as its last
vertex and $P_{2}$ is a longest normal path of $G(i+1,x-1)$ with $u_{y}$ as
its last vertex.
\end{lemma}

\begin{proof}
Note that $P$ has at least two vertices, since $u_{y},u_{i}\in V(P)$.
Therefore, since $u_{i}<_{\sigma }u_{k}$ for every $u_{k}\in V(P)\setminus
\{u_{i}\}$, it follows that $u_{i}$ is not the first vertex of $P$, 
and thus $P_{1}\neq \emptyset $. 
Note that $V(P_{1})\subseteq V(G(i+1,j))$, i.e.~$V(P_{1})\subseteq V(G(i,j))\setminus
\{u_{i}\}$ by Observation~\ref{G(i,j)-u_i}, since $u_{i}\notin V(P_{1})$.
Furthermore, since $P$ is a normal path by assumption and $P_{1}$ is a
subpath of $P$, it follows that~$P_{1}$ is a normal path of $G(i+1,j)$ with $%
u_{x}$ as its last vertex.

Let $\sigma ^{\prime }=\sigma |_{P}$ be the restriction of the ordering $%
\sigma $ on the vertices of the path $P$ and let~$\sigma ^{\prime \prime }=$%
RMN$(\sigma ^{\prime })$. Then, since $P$ is a normal path by assumption,
the ordering of the vertices of~$V(P)$ in $P$ coincides with the ordering $%
\sigma ^{\prime \prime }$ by Observation~\ref{normal-RMN}.

We will now prove that $V(P_{2})\subseteq V(G(i+1,x-1))$. Consider an
arbitrary vertex $u_{k}\in V(P_{2})$ and note that $u_{i}<_{\sigma }u_{k}$.
Note that both $u_{i}$ and $u_{k}$ are unvisited by $P$ when $u_{x}$ is
visited. Suppose that $u_{k}\in N(u_{x})$. Then, since $u_{i}<_{\sigma
}u_{k} $, it follows that $u_{i}$ is not the rightmost unvisited vertex of $%
N(u_{x})\cap V(P)$ in $\sigma $, when $P$ visits $u_{x}$. Thus, since $P$ is
normal by assumption, it follows that $u_{i}$ is not the next vertex of $%
u_{x}$ in $P$, which is a contradiction. Therefore $u_{k}\notin N(u_{x})$
for every $u_{k}\in V(P_{2})$. Recall by the previous paragraph that the
ordering of the vertices of $V(P)$ in $P$ coincides with the ordering $%
\sigma ^{\prime \prime }=$RMN$(\sigma ^{\prime })$, where $\sigma ^{\prime
}=\sigma |_{P}$. Therefore, since $u_{k}\in V(P_{2})$ appears after $u_{x}$
in $P$, it follows that $u_{x}<_{\sigma ^{\prime \prime }}u_{k}$. Thus,
since $u_{k}\notin N(u_{x})$, Lemma~\ref{auxil-2} implies that $%
u_{k}<_{\sigma ^{\prime }}u_{x}$, i.e.~$u_{k}<_{\sigma }u_{x}$. Summarizing, 
$u_{k}\notin N(u_{x})$ and $u_{i}<_{\sigma }u_{k}<_{\sigma }u_{x}$ for every 
$u_{k}\in V(P_{2})$, and thus $V(P_{2})\subseteq V(G(i+1,x-1))$ by
Definition~\ref{G(i,j)}.

Since ${u_{i}<_{\sigma }u_{x}}$ and ${u_{x}\in N(u_{i})}$, Lemma~\ref%
{u_i-G(i+1,x-1)} implies that ${u_{k}\in N(u_{i})}$ for every ${u_{k}\in
V(P_{2})\subseteq V(G(i+1,x-1))}$. Therefore, since $P=(P_{1},u_{i},P_{2})$
is a normal path by assumption, the first vertex of $P_{2}$ is the rightmost
vertex of $V(P_{2})$ in $\sigma $. Consider now any two consecutive vertices 
$u_{r},u_{r^{\prime }}$ of $P_{2}$. Then, since $P=(P_{1},u_{i},P_{2})$ is a
normal path, it follows that~$u_{r^{\prime }}$ is the rightmost unvisited
vertex of $N(u_{r})\cap V(P)$ (resp.~of of $N(u_{r})\cap V(P_{2})$) in $%
\sigma $, when $P$ (resp.~$P_{2}$) visits~$u_{r}$. Therefore, since also $%
u_{y}$ is the last vertex of $P$ by assumption, $P_{2}$ is a normal path of $%
V(G(i+1,x-1))$ with $u_{y}$ as its last vertex.

Suppose now that there exists a normal path $P_{1}^{\prime }$ (resp.~$%
P_{2}^{\prime }$) of $G(i+1,j)$ (resp.~of~$G(i+1,x-1)$) with $u_{x}$ 
(resp.~with~$u_{y}$) as its last vertex, such that $|P_{1}^{\prime }|>|P_{1}|$
(resp.~$|P_{2}^{\prime }|>|P_{2}|$). Then, Lemma~\ref{P=(P1,ui,P2)-normal}
implies that $P^{\prime }=(P_{1}^{\prime },u_{i},P_{2})$ (resp.~$P^{\prime
}=(P_{1},u_{i},P_{2}^{\prime })$) is a normal path of $G(i,j)$ with~$u_{y}$
as its last vertex, such that $|P^{\prime }|>|P|$. This is a contradiction
to the assumption that $P$ is a longest normal path of~$G(i,j)$ with $u_{y}$
as its last vertex. Therefore, there exists no such path $P_{1}^{\prime }$
(resp.~$P_{2}^{\prime }$), and thus $P_{1}$ (resp.~$P_{2}$) is a longest
normal path of $G(i+1,j)$ (resp.~of~$G(i+1,x-1)$) with $u_{x}$ (resp.~with $%
u_{y}$) as its last vertex. This completes the proof of the lemma.
\end{proof}

\subsection{The algorithm\label{algorithm-subsec}}

In the following we present our Algorithm~\ref{alg-lp} that computes a
longest path of a given cocomparability graph $G$. For simplicity of the
presentation of this algorithm, we make the following convention.

\vspace{-0.03cm}
\begin{notation}
\label{alg-notation}Let ${G=(V,E)}$ be a cocomparability graph and ${\sigma
=(u_{1},u_{2},\ldots ,u_{n},u_{n+1})}$ be an LDFS umbrella-free ordering of $%
{V\cup \{u_{n+1}\}}$, where $u_{n+1}$ is a dummy isolated vertex. For every
pair of indices~${i,j\in \{1,2,\ldots ,n\}}$ and for every vertex~${u_{k}\in
V(G(i,j))}$, we denote by~$P(u_{k};i,j)$ a longest normal path of $G(i,j)$
with $u_{k}$ as its last vertex and by $\ell (u_{k};i,j)$ the length $%
|P(u_{k};i,j)|$ of $P(u_{k};i,j)$, i.e.~the number of vertices of $P(u_{k};i,j)$.
\end{notation}

We first give a brief overview of Algorithm~\ref{alg-lp}. 
It takes as input a cocomparability graph~${G=(V,E)}$ and an umbrella-free ordering $\pi$ of~$V$. 
As a preprocessing step, the algorithm applies LDFS$^{+}$ (i.e.~Algorithm~\ref{ldfs+-alg}) 
to the ordering $\pi$ in order to compute an LDFS umbrella-free ordering $\sigma$ of $V$. 
In the sequel, the dynamic programming part of Algorithm~\ref{alg-lp} builds a 3-dimensional table 
where for every pair of indices $i,j \in\{1,2, \ldots, n\}$ and for every vertex $u_k \in V(G(i,j))$, 
the entry $P(u_k;i,j)$ stores the ordered vertices of a longest normal path of~$G(i,j)$
with $u_k$ as its last vertex; the length of this path (i.e.~$|P(u_k;i,j)|$) is stored in~$\ell(u_k; i,j)$.  
Thus a longest normal path of $G=G(1,n)$ will be stored in $P(u_k;1,n)$ for a~$u_k$ that 
maximizes $\ell(u_y;1,n)$ among all $u_y \in V$ (cf.~line~\ref{alg-lp-17}). 
Note that from the \emph{for}-loops in lines~\ref{alg-lp-3} and~\ref{alg-lp-4} 
of the algorithm and the obvious inductive hypothesis, it may be assumed during 
the $\{i,j\}$th iteration of the body of the dynamic programming (cf.~lines~\ref{alg-lp-7}-\ref{alg-lp-16}), 
that the values $P(u_{k^{\prime}}; i^{\prime},j^{\prime})$ and $\ell(u_{k^{\prime}}; i^{\prime},j^{\prime})$ 
have been correctly computed at previous iterations of the algorithm, for every~$i^{\prime} > i$.

On entry to the initialization phase for a particular $\{i,j\}$ (cf.~lines~\ref{alg-lp-7}-\ref{alg-lp-6}), 
we want initial paths that do not use vertex $u_i$ as an intermediate vertex.  
For a path with $u_y \in V(G(i+1,j))$ as its last vertex, such a path is stored in $P(u_y;i+1,j)$.  
For a path with $u_i$ itself as its last vertex, we are only interested in
the case where $u_i \in V(G(i,j))$ and, if so, we initialize $P(u_i;i,j)=(u_i)$.  

Then, we enter the induction step phase of the algorithm (cf.~lines~\ref{alg-lp-9}-\ref{alg-lp-16}) 
and determine how the entries of the table can be extended with the inclusion of vertex $u_i$ (in the case where~${u_i \in V(G(i,j))}$). 
First we note that, if a normal path $P$ of $G(i,j)$ that includes $u_{i}$ has at least two vertices, 
then $P$ must involve a vertex $u_x \in V(G(i+1,j))$ with $u_{x} u_{i} \in E$.  
For such a vertex $u_x$, there are two different roles that it can play in getting a possibly longer normal path to be stored in the table.  
First, adding the edge $u_{x} u_{i}$ to a longest normal path of $G(i+1,j)$ with~$u_x$ as its last vertex, 
might create a normal path with $u_i$ as its last vertex, which is longer than the one currently stored in $P(u_i;i,j)$.  
This situation is covered in lines~\ref{alg-lp-11}-\ref{alg-lp-12b}.  
The other role that vertex $u_x$ might play is to serve as the ``glue'' between a normal path $P_1$ of $G(i+1,j)$ 
with~$u_x$ as its last vertex and a normal path $P_2$ of $G(i+1,x-1)$ with some vertices~$u_{y^{\prime}}$ and~$u_y$ as its 
first and last vertex, respectively ($u_{y^{\prime}}$ and $u_y$ are not necessarily distinct).  
Note that these two paths would be ``glued'' together via the two 
edges $u_{x} u_{i}$ and $u_{i} u_{y^{\prime}}$.  This situation is covered in lines~\ref{alg-lp-13}-\ref{alg-lp-16}.

The next main theorem of this section proves that Algorithm~\ref{alg-lp}
computes in~$O(n^{4})$ time a longest path of a cocomparability graph with~$n$ vertices.

\begin{algorithm}[htb]
\caption{Computing a longest path of a cocomparability graph} \label{alg-lp}
\begin{algorithmic}[1]
\REQUIRE{A
cocomparability graph $G=(V,E)$ with $|V|=n$ and an umbrella-free ordering $\pi$ of~$V$}
\ENSURE{A longest path of $G$}

\medskip

\STATE{Run an LDFS$^{+}$ preprocessing step to $\pi$ to obtain the LDFS umbrella-free ordering $\sigma$} \label{alg-lp-1}
\STATE{Add an isolated dummy vertex $u_{n+1}$ to $\sigma$; denote $\sigma=\{u_{1},u_{2},\ldots,u_{n},u_{n+1}\}$} \label{alg-lp-2}

\medskip

\FOR{$i=n$ downto $1$} \label{alg-lp-3}
     \FOR{$j=i$ to $n$} \label{alg-lp-4}
     
\medskip

          \FOR{every $u_{y}\in V(G(i+1,j))$} \label{alg-lp-7}
               \STATE{$P(u_{y};i,j) \leftarrow P(u_{y};i+1,j)$; $\ell(u_{y};i,j) \leftarrow \ell(u_{y};i+1,j)$} \label{alg-lp-8} 
               \COMMENT{initialization}
          \ENDFOR
          
          \IF{$u_{i}\in V(G(i,j))$} \label{alg-lp-5}
               \STATE{$P(u_{i};i,j) \leftarrow (u_{i})$; $\ell(u_{i};i,j)\leftarrow 1$} \label{alg-lp-6} 
               \COMMENT{initialization}
          \ENDIF

\medskip
          
          \FOR{every $u_{x} \in V(G(i+1,j))$} \label{alg-lp-9}
          \vspace{0.1cm}
          
               \IF{$u_{i}\in V(G(i,j))$ and $u_{x}\in N(u_{i})$} \label{alg-lp-10}
               \vspace{0.1cm}
               
                    \IF{$\ell(u_{i};i,j) < \ell(u_{x};i+1,j) + 1$} \label{alg-lp-11}
                         \STATE{$P(u_{i};i,j) \leftarrow (P(u_{x};i+1,j),u_{i})$}  \label{alg-lp-12a}
                         \STATE{$\ell(u_{i};i,j)\leftarrow \ell(u_{x};i+1,j) + 1$} \label{alg-lp-12b}
                    \ENDIF
                    
               \vspace{0.1cm}
               
                    \FOR{every ${u_{y} \in V(G(i+1,x-1))}$} \label{alg-lp-13}
                         \IF{$\ell(u_{y};i,j) < \ell(u_{x};i+1,j) + \ell(u_{y};i+1,x-1) + 1$} \label{alg-lp-14}
                              \STATE{$P(u_{y};i,j) \leftarrow (P(u_{x};i+1,j),u_{i},P(u_{y};i+1,x-1))$} \label{alg-lp-15}
                              \STATE{$\ell(u_{y};i,j) \leftarrow \ell(u_{x};i+1,j) + \ell(u_{y};i+1,x-1) + 1$} \label{alg-lp-16}
                         \ENDIF
                    \ENDFOR
                    
               \ENDIF
          \ENDFOR
     \ENDFOR
\ENDFOR

\medskip

\RETURN{a path $P(u_{k};1,n)$ with $\ell(u_{k};1,n) = \max\{\ell(u_{y};1,n)\ | \ u_{y}\in V\}$} \label{alg-lp-17}
\end{algorithmic}
\end{algorithm}

\begin{theorem}
\label{correctness}
For a given cocomparability graph $G=(V,E)$ with $n$
vertices, Algorithm~\ref{alg-lp} computes a longest path $P$ of $G$ in $O(n^{4})$ time.
\end{theorem}

\begin{proof}
In the first line, Algorithm~\ref{alg-lp} applies an LDFS$^{+}$
preprocessing step to the given umbrella-free ordering $\pi $ of $V$. The
resulting LDFS ordering $\sigma $ is again umbrella-free~\cite{Corneil-MPC}.
In the second line, the algorithm adds a dummy isolated vertex $u_{n+1}$ to $%
\sigma $ to the right of all other vertices of $V$, i.e.~we consider without
loss of generality that ${\sigma =(u_{1},u_{2},\ldots ,u_{n},u_{n+1})}$.
Note that $\sigma $ remains an LDFS umbrella-free ordering, also after the
addition of $u_{n+1}$ to it. Furthermore, note that any longest path of $G$
is also maximal (cf.~Definition~\ref{maximal-path}). Therefore, in order to
compute a \emph{longest} path of $G$, it suffices by Theorem~\ref{normal-thm}
to compute a \emph{longest normal} path of $G$ (with respect to the ordering 
$\sigma $), i.e.~a longest path among the normal ones.

In lines~\ref{alg-lp-3}-\ref{alg-lp-16}, Algorithm~\ref{alg-lp} iterates
for every pair of indices~${i,j\in \{1,2,\ldots ,n\}}$ and computes a path $%
P(u_{k};i,j)$ and a value $\ell (u_{k};i,j)$ for every vertex~${u_{k}\in
V(G(i,j))}$. We will prove by induction on $i$ that $P(u_{k};i,j)$ is indeed
a longest normal path of $G(i,j)$ with $u_{k}$ as its last vertex and that $%
\ell (u_{k};i,j)=|P(u_{k};i,j)|$.

For the induction basis, let $i=n$; in this case also $j=n$ (cf.~line~\ref%
{alg-lp-4}). Furthermore~${u_{i}\notin N(u_{i+1})}$ for $i=n$, since $u_{n+1}
$ is an isolated vertex, and thus the algorithm executes line~\ref{alg-lp-6}. 
In this line, the algorithm computes the path $P(u_{n};n,n)=(u_{n})$,
which is clearly the only (and thus also the longest) normal path of $G(n,n)$
with $u_{n}$ as its last vertex. Then, since~$G(n+1,n)=\emptyset $
(cf.~Definition~\ref{G(i,j)}), lines~\ref{alg-lp-8} 
and~\ref{alg-lp-10}-\ref{alg-lp-16} are not executed at all. This proves the induction basis.

For the induction step, let ${i\leq n-1}$. Consider the iteration of the
algorithm for any ${j\in \{i,i+1,\ldots ,n\}}$. First, the algorithm
initializes in lines~\ref{alg-lp-7}-\ref{alg-lp-6} the values~$P(u_{k};i,j)$
and~$\ell (u_{k};i,j)$ for every $u_{k}\in V(G(i,j))$. Then, it updates
these values if necessary in lines~\ref{alg-lp-9}-\ref{alg-lp-16}. For every
vertex $u_{y}\in V(G(i+1,j))$, the induction hypothesis implies that~$%
P(u_{y};i+1,j)$ is a longest normal path of $G(i+1,j)$ with $u_{y}$ as its
last vertex and that~${\ell (u_{y};i+1,j)=|P(u_{y};i+1,j)|}$. Recall by
Observation~\ref{G(i,j)-u_i} that $G(i+1,j)=G(i,j)\setminus \{u_{i}\}$.
Therefore, for every $u_{y}\in V(G(i+1,j))$, the value $\ell (u_{y};i+1,j)$
is the greatest length of a normal path $P$ of $G(i,j)$ with $u_{y}$ as its
last vertex, such that $P$ does not include $u_{i}$. The algorithm
initializes in line~\ref{alg-lp-8} for every $u_{y}\in V(G(i+1,j))$ the
values $P(u_{y};i,j)$ and $\ell (u_{y};i,j)$ as~$P(u_{y};i+1,j)$ and $\ell
(u_{y};i+1,j)$, respectively. Furthermore, in the case where $u_{i}\in
V(G(i,j))$, the algorithm initializes in line~\ref{alg-lp-6} the values $%
P(u_{i};i,j)=(u_{i})$ and $\ell (u_{i};i,j)=1$. Otherwise, in the case where 
$u_{i}\notin V(G(i,j)) $, the algorithm does not execute line~\ref{alg-lp-6}, 
since the values $P(u_{i};i,j)$ and $\ell (u_{i};i,j)$ can not be defined
(cf.~Notation~\ref{alg-notation}).

Suppose that $u_{i}\in V(G(i,j))$; then the path $P(u_{i};i,j)$ is well
defined (cf.~Notation~\ref{alg-notation}). Recall by Observation~\ref%
{superpath-normal-obs} that for any normal path $P_{1}$ of $G(i+1,j)$ with a
vertex $u_{x}$ as its last vertex, such that $u_{x}\in N(u_{i})$, the path $%
(P_{1},u_{i})$ is a normal path of $G(i,j)$. Conversely, recall by
Observation~\ref{subpath-normal-obs} that the path $P(u_{i};i,j)\setminus
\{u_{i}\}$ (if not empty) is a normal path of~$G(i+1,j)$. Therefore, in
order to update the value of $P(u_{i};i,j)$, the algorithm correctly
computes in lines~\ref{alg-lp-11}-\ref{alg-lp-12b} the paths~$%
(P(u_{x};i+1,j),u_{i})$ for every $u_{x}\in V(G(i+1,j))$, such that $%
u_{x}\in N(u_{i})$, and keeps the longest of them.

Recall now that for every $u_{y}\in V(G(i+1,j))$, the value $\ell
(u_{y};i+1,j)$ is the greatest length of a normal path $P$ of $G(i,j)$ with $%
u_{y}$ as its last vertex, such that $P$ does not include $u_{i}$.
Furthermore, recall that for every $u_{y}\in V(G(i+1,j))$ the values $%
P(u_{y};i,j)$ and $\ell (u_{y};i,j)$ have been initialized in line~\ref{alg-lp-8} 
as $P(u_{y};i+1,j)$ and $\ell (u_{y};i+1,j)$, respectively. In
the case where $u_{i}\in V(G(i,j))$ (cf.~line~\ref{alg-lp-10}), the
algorithm executes lines~\ref{alg-lp-14}-\ref{alg-lp-16} for every $u_{x}\in
V(G(i+i,j))$ with $u_{x}\in N(u_{i})$ and for every $u_{y}\in V(G(i+1,x-1))$. 
For such a pair of vertices $u_{x},u_{y}$, recall by Lemma~\ref%
{P=(P1,ui,P2)-normal} that $(P(u_{x};i+1,j),u_{i},P(u_{y};i+1;x-1))$
is a normal path of~$G(i,j)$ with~$u_{y}$ as its last vertex. Conversely,
let $P$ be a normal path of~$G(i,j)$ with $u_{y}\neq u_{i}$ as its last
vertex, let $P=(P_{1},u_{i},P_{2})$, and let $u_{x}$ be the last vertex of $%
P_{1}$. Then Lemma~\ref{P1,P2-normal} implies that $P_{1}={P(u}_{x};i+1,j)$
and $P_{2}={P(u}_{y};i+1,x-1)$. Therefore, the algorithm correctly computes
during the multiple executions of lines~\ref{alg-lp-14}-\ref{alg-lp-16} the
greatest length $\ell $ of a normal path $P$ of $G(i,j)$ with $u_{y}$ as its
last vertex, such that $P$ includes $u_{i}$. If at least one of these paths
has greater length than the initial value $\ell (u_{y};i,j)$ that has been
computed in line~\ref{alg-lp-8}, the algorithm keeps in~$P(u_{y};i,j)$ the
longest among these paths. This completes the induction step.

Therefore, for every pair of indices $i,j\in \{1,2,\ldots ,n\}$ (such that ${%
G(i,j)\neq \emptyset }$) and every ${u_{k}\in V(G(i,j))}$, the algorithm
correctly computes after the execution of lines~\ref{alg-lp-1}-\ref%
{alg-lp-16} a longest normal path $P(u_{k};i,j)$ of $G(i,j)$ with $u_{k}$ as
its last vertex and its length $\ell (u_{k};i,j)=|P(u_{k};i,j)|$. Finally,
the algorithm computes and returns in line~\ref{alg-lp-17} the longest among
the paths $P(u_{y};1,n)$, where $u_{y}\in V(G(1,n))$. Since $G(1,n)=G$, the
returned path is a longest normal path of $G$, and thus also a longest path
of $G$ by Theorem~\ref{normal-thm}.

Before establishing the running time of the algorithm, we discuss some
implementation details.  First of all, to avoid the search of the table indicated
in line~\ref{alg-lp-17}, the length and location of the current longest path
would be maintained throughout the algorithm.  Secondly, we have to
state exactly what is stored in each entry of the table.  Following standard
dynamic programming techniques, we do not store the path itself, but
rather, an indication of how the path is built.  In particular, each of 
lines~\ref{alg-lp-8},~\ref{alg-lp-6},~\ref{alg-lp-12a}, and~\ref{alg-lp-15}
gives ``instructions'' on how to build the current longest path using 
information that has already been computed.  At the end of the
algorithm a simple recursive unwinding of these ``instructions''
yields a longest path in the given graph.

Regarding the running time of Algorithm~\ref{alg-lp}, we first examine the
dynamic programming part of the algorithm.  Lines~\ref{alg-lp-14}-\ref{alg-lp-16} 
lie in four loops of $O(n)$ iterations each.  
Following the implementation details described above, each step in
lines~\ref{alg-lp-14}-\ref{alg-lp-16} can be executed in constant time, yielding an $O(n^4)$
bound on the dynamic programming portion of the algorithm.  Since the other
parts of the algorithm, even if we have to confirm that we have an umbrella-free ordering
of $V$, can easily be implemented to run in $O(n^3)$ time, 
the total running time of Algorithm~\ref{alg-lp} is~$O(n^{4})$. 
This completes the proof of the theorem.
\end{proof}

\begin{remark}
\label{remark-simpler-interval-alg}
Recall by Observation~\ref{interval-umbrella-free-obs} that an
I-ordering $\sigma$ of any interval graph $G$ is also an umbrella-free
ordering.  Furthermore, it is easy to see that $\sigma$ is also an
LDFS ordering. Thus, since lines~\ref{alg-lp-2}-\ref{alg-lp-16} of Algorithm~\ref{alg-lp} 
are applied to such an ordering $\sigma$, and since interval graphs are strictly included in
cocomparability graphs~\cite{Brandstaedt99}, Theorem~\ref{correctness}
implies that Algorithm~\ref{alg-lp} (which is essentially simpler than the
algorithm presented in~\cite{longest-int-algo}) also computes with the same time
complexity a longest path of an interval graph.
\end{remark}

\vspace{-0.52cm}
\section{Conclusion and further research\label{conclusions}}

In this paper we provided the first polynomial algorithm for the longest
path problem on cocomparability graphs. This algorithm is based on a dynamic
programming approach that is applied to a Lexicographic Depth First Search
(LDFS) characterizing ordering of the vertices of cocomparability graphs.  Our results
provide hope that this general dynamic programming approach can be used
in a more general setting, leading to efficient algorithms for the longest
path problem on even greater classes of graphs. Furthermore, more
interestingly, in addition to the recent results presented in~\cite%
{Corneil-MPC}, our results also provide evidence that cocomparability graphs
present an interval graph structure when they are considered using 
an LDFS characterization ordering of their vertices, which may lead to other new and more efficient
combinatorial algorithms. Many interesting open questions are raised by the
results in this paper:

\begin{itemize}
\vspace{-0.16cm}
\item There are now two path problems where the interval graph algorithm can
be modified by the addition of an LDFS$^+$ preprocessing sweep to solve the
same problem on cocomparability graphs. Are there other such problems? (Note
that the Hamiltonian cycle algorithm for interval graphs does not seem to extend to 
cocomparability graphs in this way.)\vspace{-0.16cm}

\item More importantly, is there an underlying ``interval structure'' in
cocomparability graphs exposed by an LDFS$^+$ sweep of an umbrella-free ordering?\vspace{-0.16cm}

\item There are many applications of multi-sweeping of LBFS 
(see~\cite{CorneilLBFS09} for a recent result; for a survey see~\cite{Corneil-LBFS-survey}). 
Is anything gained by multi-sweeping LDFS?\vspace{-0.16cm}

\item Are there other applications of LDFS?\vspace{-0.16cm}

\item Can the new Hamiltonian path, minimum path cover, and longest path
algorithms for cocomparability graphs be extended to asteroidal triple-free
(AT-free) graphs, or failing that to graph classes that lie between
cocomparability graphs and AT-free graphs~\cite{AT-subfamilies01-SIDMA}? 
The complexity of all Hamiltonicity problems is still open for AT-free graphs.\vspace{-0.16cm}

\item Can LDFS be implemented to run in linear time?\vspace{-0.16cm}
\end{itemize}

{\small 
\bibliographystyle{abbrv}
\bibliography{ref-lcocomp}
}

\end{document}